\documentclass[11pt,a4paper]{article}
\usepackage[utf8]{inputenc}
\usepackage{amsmath}
\usepackage{amsthm}
\usepackage{amsfonts}
\usepackage{amssymb}
\usepackage{bbm}
\usepackage{natbib}
\usepackage{amsfonts}
\usepackage{rotating}
\usepackage{graphicx}
\usepackage{graphicx,color,graphics,lscape,srcltx}
\usepackage[top=1.0in, bottom=1.0in, left=1in, right=1in]{geometry}
\usepackage{float}
\usepackage{hyperref}
\usepackage{caption}
\usepackage[normalem]{ulem}
\usepackage{subfig}
\usepackage{lscape}
\usepackage{comment}

\usepackage{marginnote}
\usepackage{xcolor}

\newtheorem{defn}{Definition}

\newtheorem{thm}{Theorem}[section]

\newtheorem{prop}[thm]{Proposition}

\newtheorem{rem}{Remark}

\begin{document}

\title{Statistical description and dimension reduction of continuous time categorical trajectories with multivariate functional principal components}

\author{Hervé \textsc{Cardot}$^{1}$ and Caroline \textsc{Peltier}$^{2,3}$ \\
{\small
$^{1}$Université Bourgogne Europe, CNRS, Institut de Math\'{e}matiques de Bourgogne, UMR CNRS 5584, Dijon, France} \\
{\small $^{2}$Inrae, Centre des Sciences du Goût et de l'Alimentation, UMR CNRS-INRAE-Institut Agro, Dijon, France}  \\
{\small $^{3}$Probe Research Infrastructure, Chemosens facility, CNRS-INRAE, Dijon, France }}

\maketitle

\begin{abstract}

Getting  tools that allow simple representations and comparisons of a set of categorical trajectories is of major interest for statisticians. Without loosing any information, we associate to each state  a binary random indicator function, taking values in $\{0,1\}$, and turn the problem of statistical description of the categorical trajectories into a multivariate functional principal components analysis.  
This viewpoint encompasses experimental frameworks where two or more states can be observed simultaneously.
The sample paths being piecewise constant, with a finite number of jumps, this  a rare case in functional data analysis in which the trajectories are not supposed to be continuous and can be observed exhaustively. Under the weak hypothesis assuming only continuity in probability of the $0-1$ trajectories, the means and the (multivariate) covariance function are continuous and  have interpretations in terms of departure from independence of the joint probabilities.  Considering a functional data point of view, we show that the binary trajectories, which are right-continuous functions with left-hand limits, can be seen as random elements in the Hilbert space of square integrable functions. The multivariate functional principal components are simple to interpret and we show that we can  get consistent estimators of the mean trajectories and the covariance functions under weak regularity assumptions.  The ability of the approach to represent categorical trajectories in a small dimension space is illustrated on a data set of sensory perceptions, considering different  gustometer-controlled stimuli experiments. 
\end{abstract}

\noindent \textbf{Keywords} : categorical functional data, continuous time categorical processes, dimension reduction, discontinuous trajectories, jump processes, sensometrics, telegraph process, temporal dominance of sensations, temporal check all that apply.

\section{Introduction}

A lot of attention has been paid in the statistical literature over the last decades to develop tools dedicated to the  analysis and modeling of functional data, considering random trajectories defined on an interval  $[0,T]$ and taking values in $\mathbb{R}$, at each $t \in [0,T]$ (see for example \citealp{Ger2024} and \citealp{KonS2023} for recent reviews of the literature and \citealp{RamsayS2005} for a seminal reference on functional data analysis).  Much less attention has been given to the case in which the trajectories take values in a finite set of elements that are not necessarily numbers, that is to say when $\mathbb{R}$ is replaced by the finite set $\mathcal{S} = \{S_1, \ldots, S_q\}$, with cardinality $q$.

However, there are many examples in which the statistical units of interest are samples of a continuous time random categorical process :  in a  pioneer work on demographic studies,  \cite{DEV1982}  extended the notion of multiple correspondence analysis to continuous time correspondence analysis in order to analyze the time evolution of the ``marital status'' of a sample of women over the period $[1962,1975]$. The  trajectories related to the marital status take values in a state space $\mathcal{S}$ with $q=4$ states (divorced, married, single, widowed). Recent examples of statistical analysis of individual categorical trajectories are found in food science, a domain in which it is of great interest to get information about the temporal perception of aliments
to understand the perception mechanisms. Of particular interest is the Temporal Dominance of Sensations (TDS) approach  which consists in choosing sequentially attributes, among a list composed of $q$ predefined items, describing a food product over tasting (see \citealp{Pineau_2009}). 
 The chosen states correspond to the
most striking perception at a given time and the results of TDS experiments, after time normalization, can be represented via barplots as in Figure~\ref{fig:TDS}.  The technique developed by \citealp{DEV1982}, also called qualitative harmonic analysis (see \citealp{DevSap1980}) or categorical functional data analysis (CFDA) is now available in \textbf{R}, package \texttt{cfda} (see  \citealp{PGV2021}) and can be used to analyze such kind of data. Even if the CFDA  approach can be very powerful by  encoding categorical trajectories into a sequence of real components (see  \citealp{PVSC2023} for an illustration in sensory analysis), 
the interpretation of the results in terms of individual trajectories is often delicate since it is not the purpose of the method. 

\begin{figure}
\centering
\includegraphics[height=9cm]{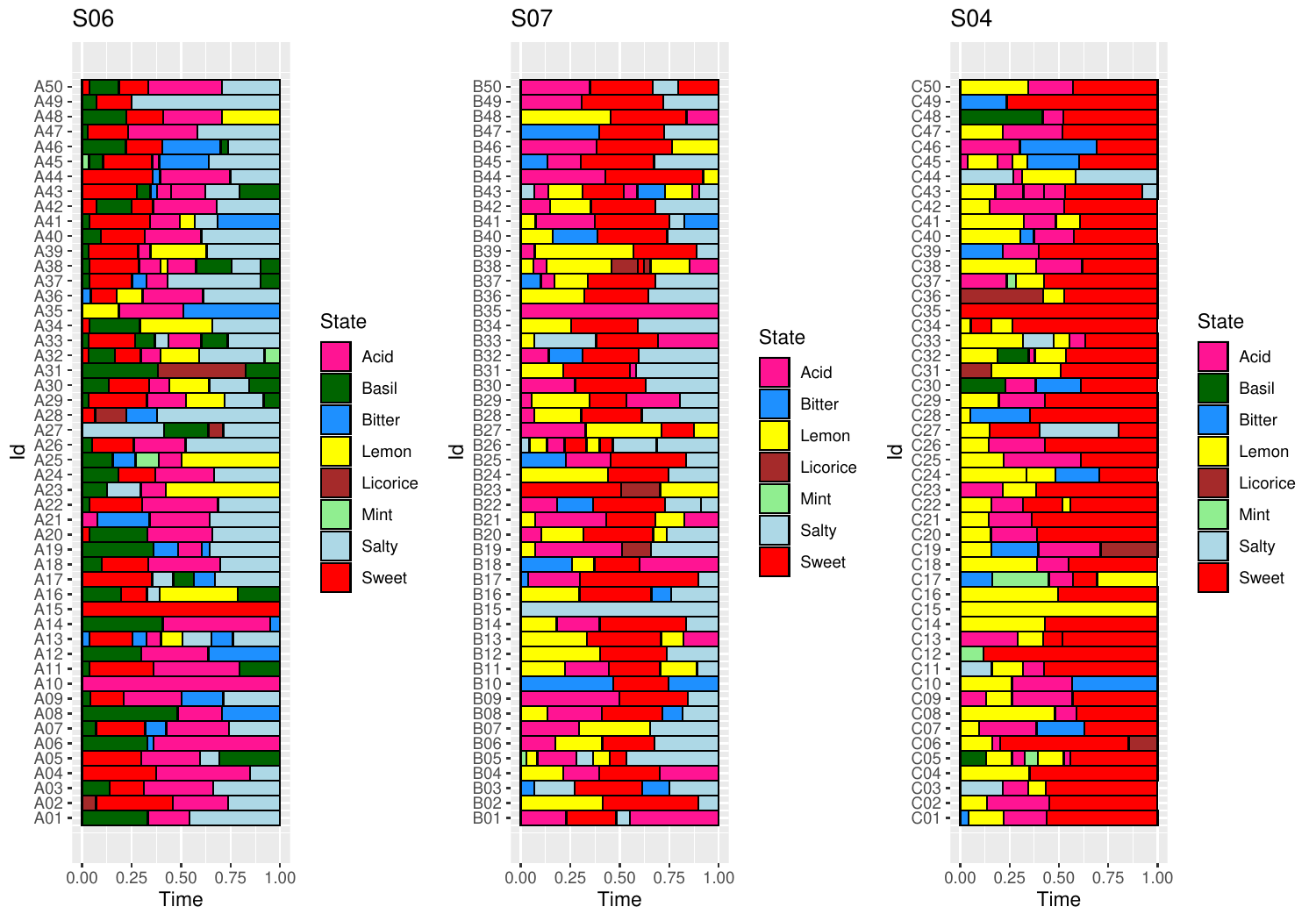}
\caption{TDS bandplot for $n=150$ tasting experiments and $q=8$ states, considering  three different  gustometer-controlled stimuli, S06, S07 and S04, extracted from the open data basis~\cite{BNV2023}. Each row corresponds to a categorical trajectory.}
\label{fig:TDS}       
\end{figure}

Statistical approaches based on Markov processes and their extensions (see \citealp{Lindsey} for an overview and \citealp{LO2001} for an introduction to semi-Markov processes) can be useful to fit the law of the  trajectories at the population level  and to provide a simple representation of the dynamics via the graph of transitions between states. Considering parametric distributions for the sojourn time in the different states also allows to deal with maximum likelihood estimation techniques, permitting to build two-sample tests to compare two populations, \textit{e.g} two food products or two categories of consumers (see \citealp{CF2024}) as well as model-based clustering, considering mixtures of semi-Markov processes  (see \citealp{CFSV2019}).  
A major drawback of this approach is that it is not well suited for analyzing data at the individual level, and the Markovian assumption is often too simplistic to properly fit real data. Additionally, it is difficult to apply when more than one state can be present simultaneously (see for example the TCATA experiment described in \citealp{CAS2016} and in the Supplementary Material), as this requires drastically increasing the number of system states to $2^q$ to account for all possible combinations.

In this work, we  introduce another point of view and associate to each state $S_j$, for $j=1, \ldots, q$,   a random trajectory $X_j(t), t \in [0,T]$, taking value 1 when state $S_j$ is observed at time $t$ and zero else.  The information given by a  categorical trajectory $Y$  is equivalent to the information given by the  $q$ binary 0-1 trajectories $X_1, \ldots, X_q$. 

One could consider such trajectories as compositional data evolving over time (see \citealp{Ait83} for a seminal paper and  \citealp{Greenacre2021} for a recent review on compositional data analysis).  A major difficulty to deal with compositional data approaches in our framework is due to the fact that we deal with individual trajectories so that, at each instant $t$,  $q-1$ individual observations, among  $X_1(t), \ldots, X_q(t)$, have value 0. In other words, at each instant $t$, all the observed units are on the vertices of the simplex, so that usual logarithmic transforms cannot apply directly.

The approach presented in this work is an extension of \cite{CP2025b}, giving a more detailed mathematical description of the properties of the trajectories and estimators of the mean and the covariance. 
This work encompasses both the TDS and TCATA experimental frameworks, the latter allowing for the simultaneous observation of two or more states.

Under the assumption that the trajectories are continuous in probability and have almost surely a finite number of jumps
we get that such random elements are jointly measurable (see \citealp{HE2015}, Chapter 7) and  we deduce that they can be seen as random elements in a multidimensional Hilbert space of square integrable functions. 
This leads us to consider the multivariate functional principal component analysis (MFPCA) of the multidimensional  functional vector $\mathbf{X}=(X_1, \ldots, X_q)$.  
When having a sample of trajectories at hand, consistency results are obtained both from a functional data  and continuous time stochastic process point of view.

We  show on the example presented in Figure~\ref{fig:TDS} that it leads both to  principal components that can be simply interpreted in terms of variations around the mean probability curve related to each state, and powerful tools able to reduce effectively the dimension of categorical functional data in a finite dimension vector space.  
We also present a comparison, on this case study, with  the CFDA approach.

All the codes in the R language (see \citealp{R2024}) are available on Github, {\small \url{https://github.com/Chemosens/ExternalCode/tree/main/MFPCAWithCategoricalTrajectories}}.
Additional numerical experiments, including weighted MFPCA, TCATA analysis, and jump time measurements corrupted by noise, along with an example based on the telegraph process, where explicit calculations are derived, are provided in the Supplementary Material.

\section{Mathematical framework}

Consider a probability space $(\Omega, \mathcal{F}, \mathbb{P})$
and suppose that the collection of  random variables $Y(t)$,  taking values in $\mathcal{S} = \{S_1, \ldots, S_q\}$, are well defined for each $t$ in $[0,T]$. We can consider, in a equivalent way, the collection of $q$-dimensional random vectors $\mathbf{X}(t) = (X_1(t), \ldots, X_q(t))$, where  $X_j(t)$, $j=1, \ldots, q$, are binary stochastic processes, related to the occurrence of state $S_j$ at time $t$, and defined as follows. For all $\omega \in \Omega$, $X_j(t,\omega) = 1$ if $Y(t,\omega) = S_j$ and $0$ otherwise. In other words, $X_j(t) = \mathbf{1}_{\{Y(t)=S_j\}}$, for all $t \in [0,T]$, where $ \mathbf{1}_{\{ . \}}$ is the indicator function. 

As  in \cite{Bill1999} or \cite{HE2015}, it is possible consider two distinct point of views to  describe mathematically  the continuous time categorical trajectories drawn in Figure~\ref{fig:TDS}. 
On the one hand we can consider the continuous time stochastic process point of view and study the properties of $\mathbf{X}(t)$  defined for each $t$ in $[0,T]$. On the other hand  we  can consider a description based on a functional point of view by studying  the random element $\mathbf{X}(.,\omega) = \{ (X_1(t,\omega), \ldots, X_q(t,\omega)) , \ t \in [0,T] \}$, taking values in a function space,  for each $\omega \in \Omega$.

Note that in TDS experiments, where exactly one state is selected at time $t$, the multivariate stochastic process $\mathbf{X}(t)$ follows a multinomial distribution with $n = 1$ trial and a success probability vector $\mathbf{p}(t) = (p_1(t), \ldots, p_q(t))$.
In TCATA experiments, however, the values of $Y(t, \omega)$ are no longer single elements of $\mathcal{S}$ but subsets of $\mathcal{S}$. We define  
$X_j(t) = \mathbf{1}_{\{Y(t) \ni S_j\}}$, which takes the value 1  if the state $S_j$ is observed at time $t$, that is to say $S_j \in Y(t,\omega)$, and 0 otherwise. We have  $\sum_{j=1}^q X_j(t) \in \{0, \ldots, q\}$  and the joint distribution of $\mathbf{X}(t)$ is no longer multinomial. Nevertheless, the marginal distribution of $X_j(t)$ remains Bernoulli with success probability denoted by $p_j(t)$.

Unless explicitly stated otherwise, we will assume from now on that $\mathbf{X}(t)$ is the result of a TDS or TCATA type experiment.

\subsection{The continuous time stochastic process point of view}

 We denote by  $p_{j}(t) = \mathbb{P}\left[Y(t) \ni S_j\right]$, the probability of observing state $S_j$ at time $t$ and, for  $(s,t) \in [0,T] \times [0,T]$ and $(j,\ell) \in \{1, \ldots, q\} \times \{1, \ldots, q\}$,  the joint probability  
\[
p_{j \ell} (s,t) = \mathbb{P}\left[Y(s) \ni S_j \mbox{ and } Y(t) \ni S_\ell \right].
\]
We clearly have 
\begin{align*}
\mathbb{E}[X_j(t)] &= \mathbb{E}[\mathbf{1}_{\{Y(t) \ni S_j\}}] \\
 &= p_j(t).
\end{align*}
Consider two distinct time points $(s,t)$, we have
 \begin{align*}
 \mathbb{E}[X_j(s)X_\ell(t)] &= \mathbb{E}[\mathbf{1}_{\{Y(s) \ni S_j\}} \mathbf{1}_{\{Y(t) \ni S_\ell\}}] \\
  &= p_{j \ell} (s,t).
 \end{align*}
We suppose in the following  that hypothesis $\mathbf{H}_1$ is satisfied  
\begin{align*}
\mathbf{H}_1 : & \quad\lim_{h \to 0} \mathbb{P}[Y(t) \neq Y(t+h) ] = 0, \quad  \forall t \in [0,T]
\end{align*}
ensuring that the trajectories of $X_j, \ j=1, \ldots, q$ are continuous in probability. Hypothesis $\mathbf{H}_1$  prevents  them to have too many jumps or to have jumps that occur at a same time point $t_0 \in [0,T]$ with non null probability.  
Define the covariance functions $\gamma_j(s,t) = \mathbb{C}ov (X_j(s),X_j(t))$ and, for $\ell \neq j$, $\gamma_{j \ell} (s,t) = \mathbb{C}ov (X_j(s),X_\ell(t))$. We remark that $\gamma_{j}(t,t) = p_{j}(t) (1 - p_j(t))$ and  $\gamma_{j\ell}(s,t) = p_{j \ell} (s,t) - p_j(s) p_\ell(t)$, so that the terms out of the diagonal, when $s \neq t$, can be related to a departure from independence. 
 
\begin{prop} Under hypothesis $\mathbf{H}_1$, we have for all $j  \in \{1, \ldots, q\}$, 
\begin{itemize}
\item $p_j$ (resp. $\gamma_j$) is continuous on $[0,T]$ (resp.  on $[0,T] \times [0,T]$)
\item $\gamma_{j \ell}$ is  continuous on $[0,T] \times [0,T]$, for all $\ell \neq j$.
\end{itemize}
\label{prop1}
\end{prop}

\begin{proof} 
{\em  
Remarking that $\sup_{t \in [0,T]} |X_j(t)| \leq q$ almost surely, we deduce, by Theorem 1.3.6 in \cite{Serfling1980}, that the trajectories are also continuous in the $L^2$ sense  (or mean square continuous) 
 when $\mathbf{H}_1$ is true, that is to say 
\begin{align}
\lim_{h \to 0} \mathbb{E}\left[ \left(X_j(t) - X_j(t+h) \right)^2\right] = 0,  \quad 
\forall t \in [0,T], \ \forall j \in \{1, \ldots, q\}.
\label{qmcont}
\end{align}

The continuity of $p_j$ and  $\gamma_{jj}$ is then a consequence of  Theorem 7.3.2 in~\cite{HE2015} which states that the mean $p_j$ and covariance functions $\gamma_{jj}$ are continuous if and only if $X_j$ is mean-square continuous. To prove the continuity of $\gamma_{j \ell}$, we note that for $(s,t)$ and $(s',t')$ in $[0,T]\times[0,T]$, we get thanks to the Cauchy-Schwarz inequality,
\begin{align*}
\left| \gamma_{j\ell}(s,t) - \gamma_{j\ell}(s',t') \right| & \leq  \left| \mathbb{C}ov (X_j(s) - X_j(s'),X_\ell(t)) \right| + \left| \mathbb{C}ov (X_j(s'),X_\ell(t) - X_\ell(t')) \right| \\
 & \leq \sqrt{ \gamma_{\ell \ell}(t) \mathbb{E}\left[ \left( X_j(s) - X_j(s')\right)^2\right]}  + \sqrt{ \gamma_{j j}(s) \mathbb{E}\left[ \left( X_\ell(t) - X_\ell(t')\right)^2\right]} 
\end{align*}
and we conclude using \eqref{qmcont}.
}
\end{proof}

Although stochastic processes with values in $\{0, 1\}$ can exhibit very smooth mean functions (see Appendix C in the Supplementary Material), their trajectories, being discontinuous due to jumps of size 1, are not mean square differentiable. This lack of regularity translates into non-differentiability of the covariance function near the diagonal $(t, t)$, as discussed in Chapter 11, \cite{Loeve1978}.

As in \cite{HE2015}, Chapter 7, we define the cross-covariance operator between $X_j$ and $X_\ell$ as the integral operator $\Gamma_{j \ell} : L^2[0,T] \to L^2[0,T]$, such that for all $f \in L^2[0,T]$,
\begin{align}
\Gamma_{j \ell} f(s) = \int_0^T \left( p_{j \ell} (s,t) - p_j(s) p_\ell(t)  \right) f(t) dt, \quad \forall s \in [0,T].
\label{defGammajl}
\end{align}
Note that since $p_{j \ell} (s,t) - p_j(s) p_\ell(t)$ is a continuous function,  $\Gamma_{j \ell} f$ is also a continuous function (see Lemma 4.6.1 in \citealp{HE2015}).

Let us introduce some notation. Consider the separable Hilbert space $L^2[0,T]$ 
equipped with the usual inner product, denoted by $\langle \cdot, \cdot \rangle$, 
and norm $\|\cdot\|$. We further define the vector space 
$\mathbb{H} = L^2([0,T], \mathbb{R}^q)$ as the separable Hilbert space of functions $\mathbf{f}= (f_1, \ldots, f_q)$ defined on $[0,T]$, with values in  $\mathbb{R}^q$ such that $\| \mathbf{f} \|_{\mathbb{H}}^2 = \sum_{j=1}^q \| f_j \|^2 < \infty$. For any $\mathbf{u} = (u_1, \dots, u_q)$ and 
$\mathbf{v} = (v_1, \dots, v_q)$ in $\mathbb{H}$, the inner product is 
defined component-wise as follows,
\begin{equation}
    \langle \mathbf{u}, \mathbf{v} \rangle_{\mathbb{H}} = \sum_{j=1}^q \langle u_j, v_j \rangle.
\label{psH}    
\end{equation}


 We denote by $\boldsymbol{\Gamma} : \mathbb{H} \to \mathbb{H}$ the multivariate integral  covariance operator as follows. For all $\mathbf{f} = (f_1, \ldots, f_q)  \in \mathbb{H}$,
\begin{align}
\boldsymbol{\Gamma} \mathbf{f}
 & = \sum_{\ell=1}^q 
 \begin{pmatrix}  \Gamma_{1\ell} f_\ell \\ \vdots \\  \Gamma_{q \ell} f_\ell \end{pmatrix}.
 \label{Gammawstp}
 \end{align}

The operator $\boldsymbol{\Gamma}$ can be seen as the integral operator with multivariate kernel function $\boldsymbol{\gamma}(s,t)$, with elements $\left( \gamma_{j \ell}(s,t) \right)_{j, \ell}$.
We deduce from Proposition~\ref{prop1} and the multivariate version of Mercer's theorem (see \citealp{CCY2014} or \citealp{HG2018}) that
there exists a set of continuous and orthonormal basis functions $\boldsymbol{\phi}_r = (\phi_{r1}, \ldots, \phi_{rq}) \in \mathbb{H}$, $r=1, 2, \ldots$ and corresponding eigenvalues $\lambda_1 \geq \lambda_2 \geq \cdots \geq 0$ such that
\begin{align}
\boldsymbol{\gamma}(s,t) &= \sum_{r \geq 1} \lambda_r \boldsymbol{\phi}_r(s) \boldsymbol{\phi}_r(t)^\top, \quad \forall (s,t) \in [0,T] \times [0,T],
\label{def:gammaMercer}
\end{align} 
where the sum converges absolutely and uniformly on $[0,T] \times [0,T]$.
This  leads to the following expansion of the covariance functions $\gamma_{j\ell}(s,t)$,
\begin{align}
  p_{j \ell} (s,t) - p_j(s) p_\ell(t)  &= \sum_{r \geq 1} \lambda_r \phi_{rj}(s) \phi_{r \ell}(t), \quad \forall (s,t) \in [0,T] \times [0,T],
 \label{KLgammajl}
\end{align} 
where the infinite sum converges uniformly on $[0,T] \times [0,T]$.

\begin{rem}\label{rem:rk1} The continuous time extension of correspondence analysis, named qualitative harmonic analysis and  developed by \cite{DEV1982} and \cite{DevSap1980} is based on the eigendecomposition of another integral operator, with a purpose that is not to expand the trajectories themselves in an ``optimal" way but to relate the states $S_j$  to numerical values at each instant $t$.
More precisely, the aim is to find an optimal encoding function $\varphi : \mathcal{S} \times [0,T] \to \mathbb{R}$ 
minimizing  the following criterion 
\begin{equation}
  \int_0^T \! \! \! \! \int_0^T \mathbb{E} \left[    \left( \varphi(Y(t),t) - \varphi(Y(s),s) \right)^2 \right]  ds  dt,
\label{def:ecodeDeville}  
\end{equation}
subject to identifiability constraints   $\mathbb{E}  \left[   \varphi(Y_t,t) \right] = 0$ for all $t \in [0,T]$
and unit variance, \\ $\int_{0}^T \mathbb{E}  \left[    \varphi(Y(t),t)^2 \right] dt  = 1$.
A solution $\varphi(x,t)$ satisfies, for all $x \in \mathcal{S}$, the integral operator equation (see equation (42) in \citealp{DEV1982}),
\begin{align}
 \sum_{\ell = 1}^q \int_0^T  \frac{p_{j \ell }(t,s)}{p_j(t)p_\ell(s)} \varphi(S_\ell,s) p_\ell(s) ds &=  \lambda \varphi(x,t), \quad \forall t \in [0,T].
 \label{CFDAdev82}
\end{align}
Departure from independence is evaluated via the ratio $\frac{p_{j \ell}(t,s)}{p_j(t)p_\ell(s)}$ which is equal to one in case of independence. Denoting by $\widetilde{\lambda}_1 \geq \widetilde{\lambda}_2 \geq \ldots \geq 0$ the sequence eigenvalues of operator equation~\eqref{CFDAdev82} and by $\boldsymbol{\varphi}_i(t) = (\varphi_i(S_1,t), \ldots,  \varphi_i(S_q,t) )$ the eigenfunction related to $\widetilde{\lambda}_i$, we get  a Mercer type expansion of the joint probabilities (see equation (43) in \citealp{DEV1982}),
\begin{align}
p_{j \ell}(t,s) &= p_j(t)p_\ell(s) \left( \sum_{r \geq 1} \widetilde{\lambda}_i \varphi_i(S_j,t) \varphi_i(S_\ell,s) \right). 
\label{cfda:kl}
\end{align}

This means that the optimal encoding approach considers implicitly  a multiplicative point of view to expand the departure from independence of the joint probabilities $p_{j \ell}(s,t)$ whereas the optimal trajectory expansion studied in this article considers an additive point of view, as seen  in \eqref{KLgammajl}.
\end{rem}

\begin{rem}
\label{rem:2} 
Note that it is possible to introduce positive  weights $w_1, \ldots, w_q$ in the inner product defined in \eqref{psH}:
\begin{equation*}
    \langle \mathbf{u}, \mathbf{v} \rangle_{\mathbb{H}} = \sum_{j=1}^q w_j \langle u_j, v_j \rangle.
\end{equation*}

These weights  can be chosen by the statistician  in order to  give the same importance to all the states, by considering for example
\begin{align}
w_j &= \frac{1}{\mbox{tr} \left( \Gamma_{jj} \right) }  \nonumber \\
 &= \left( \int_0^T p_j(t) \left( 1 - p_j(t) \right) dt \right)^{-1}, \label{weightstoone}
\end{align}
 setting to one the trace of the covariance operator $(w_j)^{-1} \Gamma_{jj}$ of the binary (normalized) trajectory  $X_j$ (see  the supplementary material for an illustration on the sensory data).
\end{rem}

\subsection{The functional point of view and multivariate PCA}
We now adopt another point of view and introduce the space $D([0,T], \mathbb{R}^q)$, a multidimensional version of the space $D([0,T])$,  the set of functions which have only jump discontinuities and are right continuous (see Chapter 14 in \citealp{Breiman1968} or Chapter 3 in \citealp{Bill1999}).

\begin{defn}
Let $D([0,T], \mathbb{R}^q)$ be the set of functions $\mathbf{x} : [0,T] \mapsto \mathbb{R}^q$ that are right-continuous and have left-hand limits,
\begin{itemize}
    \item for $0 \leq t < T$, $\mathbf{x}(t+) = \lim_{s \searrow t}\mathbf{x}(s)$ exists and $\mathbf{x}(t+)=\mathbf{x}(t)$,
    \item for $0 < t \leq T$, $\mathbf{x}(t-) = \lim_{s \nearrow t}\mathbf{x}(s)$ exists.
\end{itemize}
\end{defn}

As stated in  \cite{Breiman1968}, an element $\mathbf{x}$ of $D([0,T], \mathbb{R}^q)$ has at most countably many jumps, and for any $\delta >0$ and the set $\{ t \in [0,T] ; \| \mathbf{x}(t) - \mathbf{x}(t-) \| \geq \delta \}$ is finite, where $\| .\|$ is a norm on $\mathbb{R}^q$.

The following proposition allows to get a joint measurability result under the weak hypothesis that the trajectories belong to $D([0,T], \mathbb{R}^q)$.
Note that it is a general statement, that does not assume that the trajectories are drawn from TDS or TCATA experiments. It can be seen as a generalization of Theorem 7.4.2 in \cite{HE2015} which replaces the hypothesis that the trajectories $X(.,\omega)$ are continuous for all $\omega \in \Omega$ by the slightly weaker assumption that they belong to $D([0,T], \mathbb{R}^q)$.

\begin{prop}\label{prop:jointm}
Let $(\Omega,\mathcal F,\mathbb P)$ be a probability space and let $\mathbf{Z} : [0,T]\times\Omega \rightarrow \mathbb R^q$
be a stochastic process such that
\begin{itemize}
  \item for every $t\in[0,T]$, the random variable
  $\omega\mapsto \mathbf{Z}(t,\omega)$ is $\mathcal F$-measurable,
  \item for every $\omega\in\Omega$, the trajectory $\mathbf{Z}(.,\omega) : 
  t \in [0,T] \mapsto \mathbf{Z}(t,\omega)$ belongs to $D([0,T], \mathbb{R}^q)$.
\end{itemize}
Then $\mathbf{Z}$ is jointly measurable, that is,
\[
\mathbf{Z} : ([0,T]\times\Omega,\ \mathcal B([0,T])\otimes\mathcal F)
\rightarrow (\mathbb R^q,\ \mathcal B(\mathbb R^q))
\]
is measurable.
\end{prop}

\begin{proof} {\em The proof follows the same steps as the proof of Theorem 7.4.2 in \cite{HE2015}, which is stated for continuous trajectories. 
Let $(q_k)_{k\ge 1}$ be an increasing enumeration of
$\mathbb Q\cap[0,T]$ with $q_1=0$ and $q_k\uparrow T$.
For each $n\ge1$, choose a rational partition
\[
0=q_1<q_2<\cdots<q_{N_n}=1
\quad \text{such that} \quad
\max_{1\le k\le N_n-1} (q_{k+1}-q_k)\le 2^{-n},
\]
and consider the approximation process
\[
\mathbf{Z}_{(n)}(t,\omega) = \sum_{k=1}^{N_n-1} \mathbf{Z}(q_k,\omega)\,\mathbf 1_{[q_k,q_{k+1})}(t)
+ \mathbf{Z}(T,\omega)\,\mathbf 1_{\{T\}}(t).
\]

For each $k$, the mapping $(t,\omega)\mapsto \mathbf{Z}(q_k,\omega)\,\mathbf 1_{[q_k,q_{k+1})}(t)$
is measurable with respect to
$\mathcal B([0,T])\otimes\mathcal F$, since
$t\mapsto \mathbf 1_{[q_k,q_{k+1})}(t)$ is Borel measurable and
$\omega\mapsto \mathbf{Z}(q_k,\omega)$ is $\mathcal F$-measurable.
Hence $\mathbf{Z}_{(n)}$ is jointly measurable.

Fix $(t,\omega)\in[0,T]\times\Omega$.
There exists an index $k_n$ such that
$q_{k_n}\le t<q_{k_n+1}$.
By right-continuity of the trajectories belonging to $D([0,T], \mathbb{R}^q)$,
$t\mapsto \mathbf{Z}(t,\omega)$, we obtain, as $n$ tends to infinity,
$\mathbf{Z}(q_{k_n},\omega)\rightarrow \mathbf{Z}(t,\omega)$.
Therefore,  
\[
\mathbf{Z}_{(n)}(t,\omega)\rightarrow \mathbf{Z}(t,\omega).
\]
Since the pointwise limit of a sequence of measurable functions
with values in $\mathbb R^q$ is measurable,
it follows that $\mathbf{Z}$ is jointly measurable.}
\end{proof}

Under hypothesis $\mathbf{H}_1$ the components $p_j$ are continuous functions and $\mathbf{p} =(p_1, \ldots, p_q)$ belongs to $\mathbb{H}$. 
Next proposition states that the continuous time stochastic process $\mathbf{X}(t)$, that is continuous in probability and has its trajectories in $D([0,1], \mathbb{R}^q)$, can also be seen as a random element $\mathbf{X}$ taking values in $\mathbb{H}$. Furthermore, it is possible to identify,  its expectation in $\mathbb{H}$ with $\mathbf{p}$ and its covariance operator with the integral operator $\boldsymbol{\Gamma}$ defined in \eqref{Gammawstp}.
We denote by $\left( \boldsymbol{u} \otimes \boldsymbol{v} \right) \boldsymbol{\phi}= \langle  \boldsymbol{u}, \boldsymbol{\phi} \rangle_{\mathbb{H}} \boldsymbol{v}$ the tensor product of two elements of $\boldsymbol{u}$ and $\boldsymbol{v}$ of $\mathbb{H}$.

\begin{prop} Under hypothesis $\mathbf{H}_1$ and if, for all $\omega \in \Omega$, $\mathbf{X}(.,\omega)$ belongs to $D([0,T], \mathbb{R}^q)$, we have
\begin{itemize}
    \item $\mathbb{E}(\mathbf{X}) = \mathbf{p}$ in $\mathbb{H}$,
    \item $\mathbb{E} \left[ \left(\mathbf{X} -\mathbf{p}\right) \otimes \left(\mathbf{X} -\mathbf{p}\right) \right] = \boldsymbol{\Gamma}$. 
\end{itemize}
\label{def:trajfda}
\end{prop}

\begin{proof}
{\em 
First note that $D([0,T], \mathbb{R}^q) \subset \mathbb{H}$. By Proposition~\ref{prop:jointm}, $\mathbf{X}(t,\omega)$ is jointly measurable, and thus  $\mathbf{X}$ is a random element in $\mathbb{H}$. The result is then a consequence of hypothesis $\mathbf{H}_1$ and Theorem 7.4.3 in \cite{HE2015}.}
\end{proof}

 
We get with Proposition~\ref{def:trajfda}  that the operator $\boldsymbol{\Gamma}$, defined in \eqref{Gammawstp}, is the integral operator with multidimensional kernel function $\boldsymbol{\gamma}(s,t) = ( \gamma_{j \ell}(s,t))_{j,\ell=1, \ldots ,q}$. For $r \ge 1$, $\boldsymbol{\phi}_r$ is an eigenvector associated to the non negative eigenvalue $\lambda_r$,    $\boldsymbol{\Gamma} \boldsymbol{\phi}_r = \lambda_r \boldsymbol{\phi}_r$.
Furthermore, the total variance  satisfies
\[
\text{tr}(\boldsymbol{\Gamma}) = \sum_{r \geq 1} \lambda_r = \sum_{j=1}^q \int_{0}^T \gamma_{j,j}(t,t) dt = \mathbb{E} \left[ \| \mathbf{X} - \mathbf{p} \|_{\mathbb{H}}^2  \right] < \infty.
\]

The optimal linear expansion of $\mathbf{X} - \mathbf{p}$ in a $k$ dimensional vector space of $\mathbb{H}$, in terms of quadratic mean, is given by the truncated Karhunen-Loève expansion $\widetilde{\mathbf{X}}_k$ of $\mathbf{X}$,
\begin{align}
\widetilde{\mathbf{X}}_k(t) &= \mathbf{p} + \sum_{r=1}^k \langle  \mathbf{X} - \mathbf{p},  \boldsymbol{\phi}_r \rangle_{\mathbb{H}} \boldsymbol{\phi}_r(t), \quad \forall t \in [0,T].
\label{KLXk}
\end{align}

Our aim is to estimate $\mathbf{p}$ and $\boldsymbol{\phi}_r$, for $r=1, \ldots, k$, in order to be able to capture the main variations of a categorical random function $Y$ in a small $k$ dimensional vector space.


\section{Sampling and estimators} 

Suppose now we have a sample  of $n$ independent and identically distributed categorical stochastic processes $Y_1, \ldots, Y_n$ observed over $[0,T]$ and taking values in $\mathcal{S}= \{ S_1, \ldots, S_q \}$. 
For $i=1, \ldots, n$, we denote by $\mathbf{X}_i = (X_{i1}, \ldots, X_{iq})$ the vector of the $q$ corresponding binary trajectories, with $X_{ij}(t) = \mathbf{1}_{\{Y_i(t) \ni S_j\}}$. These trajectories are piecewise constant, with a finite number of jumps, and belongs to $D([0,T],\mathbb{R}^q)$. It is  a rare case in functional data analysis in which the trajectories can be observed exhaustively, that is to say for all time point $t \in [0,T]$.
As in Section 2, we can adopt two different perspectives to study properties of estimators based on $\mathbf{X_1}, \ldots, \mathbf{X}_n$, that of continuous time stochastic processes and that of functional data analysis.

\subsection{Stochastic process point of view}

We define, for all  $t$ in $[0,T]$, the empirical probabilities of occurrence  
\begin{align}
 \widehat{p}_j(t) &= \frac{1}{n} \sum_{i=1}^n X_{ij}(t),
 \label{def:empm}
\end{align}
and $\widehat{\boldsymbol{p}}(t) = (\widehat{p}_1(t), \ldots, \widehat{p}_q(t))$. We also define, for all $s$ and $t$ in $[0,T]$, the joint empirical probabilities 
\[
\widehat{p}_{j \ell}(s,t) = \frac{1}{n} \sum_{i=1}^n X_{ij}(s) X_{i\ell}(t).
\]
The estimators of the covariance functions are $\widehat{\gamma}_{j}(t,t) = \widehat{p}_j(t)(1 - \widehat{p}_j(t))$
and 
\begin{align}
\widehat{\gamma}_{j\ell}(s,t) &= \widehat{p}_{j \ell} (s,t) - \widehat{p}_j(s) \widehat{p}_\ell(t).   
 \label{def:empG}
\end{align}

Note that even $p_j$ and $\gamma_{j \ell}$ are continuous functions under hypothesis $\mathbf{H}_1$, the estimators $\widehat{p}_j(t)$ and $\widehat{\gamma}_{j\ell}(s,t)$ have discontinuity points. As stated in the following property, the estimators $\widehat{p}_j(t)$ and $\widehat{\gamma}_{j\ell}(s,t)$ converge pointwise almost surely to $p_j(t)$ and $\gamma_{j\ell}(s,t)$

\begin{prop} Suppose that $H_1$ is true. As $n$ tends to infinity, we have for all $j \in \{ 1, \ldots, q \}$ and all $t \in [0,T]$, $\widehat{p}_j(t) \to p_j(t)$ almost surely. 

We also have, for all $\ell \in \{ 1, \ldots, q\}$ and all $s \in [0,T]$, $\widehat{\gamma}_{j\ell}(s,t) \to \gamma_{j\ell}(s,t)$ almost surely.
\label{propast}
\end{prop}
The proof of Proposition~\ref{propast} follows directly from the continuous mapping theorem and the strong law of large numbers for independent and identically distributed bounded random variables and is therefore omitted.

We can also get uniform convergence of the estimators of the mean functions $p_j(.)$  using similar arguments as \cite{YaoMullerWang2005}, \cite{LiHsing2010} or \cite{HE2015}, under additional smoothness assumptions on the mean trajectory. 
Furthermore, it follows that observing the trajectories for all $t \in [0,T]$ 
is not strictly required. Instead, a sufficiently dense grid of $r_n$ 
discretization points, combined with a nonparametric smoothing procedure allow to get parametric rates of converge,  up to a logarithmic factor,  according to uniform convergence, for the mean  under smoothness regularity conditions.

We denote by $T$ a random variable  with values in $[0,T]$ and strictly positive density function $f_T(.)$ on its support $[0,T]$. We first draw randomly $r_n$ observational points, $T_1, \ldots, T_{r_n}$ in $[0,T]$, by considering $r_n$ independent copies of $T$. 
We denote by $X_{ij,k}$ the value of $X_{ij}(T_k)$ for  $j=1, \ldots, q$, $i=1, \ldots, n$ and $k=1, \ldots, r_n$. We observe, for $j=1, \ldots, q$ and $i=1, \ldots, n$,
\begin{equation}
    X_{ij,k} = p_j(T_k) + \left( X_{ij}(T_k) - p_j(T_k) \right), \quad k=1, \ldots, r_n .
    \label{eq:discret}
\end{equation}

Estimation of $p_j(t)$ can be performed via local  linear smoothers, as in \cite{LiHsing2010}. Let $K(.)$ be a symmetric probability density function on $[-1,1]$ and $K_h(t)= \frac{1}{h} K\left( \frac{t}{h} \right)$, where $h>0$ is the bandwidth. An estimator of the mean function $p_j(t)$ is given by $\widehat{p}_j(t) = \widehat{a}_{j,0}$, where
\begin{align}
(\widehat{a}_{j,0},\widehat{a}_{j,1}) &= \arg \min_{a_0,a_1} \frac{1}{n r_n} \sum_{k=1}^{r_n} \sum_{i=1}^{n} \left[X_{ij,k}  - a_0 - a_1(T_k - t) \right]^2 K_{h_p} \left( T_k - t\right). 
\label{est:loclinpj}
\end{align}

Note that in our specific context, and in contrast to \cite{YaoMullerWang2005}, the observed values in \eqref{eq:discret} are not corrupted by a random noise $\epsilon$. 
The following property is a direct consequence of Theorem 3.1 and Corollary 3.2 
in \cite{LiHsing2010}. We omit the proof, noting that the moment conditions (C5) are satisfied since $|X_{ij,k}| \leq 1$.

\begin{prop}
\label{propunifpg}
Suppose that $0<m_T \leq f_T(t) \leq M_T < \infty$ for all $t \in [0,T]$ and that $f_T$ is differentiable with a bounded derivative. Suppose that the kernel function $K(.)$ is a symmetric probability function on $[-1,1]$, and is of bounded variation on $[-1,1]$. We also assume that $p_j(.)$ is twice differentiable and the second derivative is bounded on $[0,T]$. \\
Then if $(h_p r_n)^{-1} = O(1)$ and $h_p = O \left( (\log n / n)^{1/4} \right)$ as $n$ tends to infinity, we have
\[
\sup_{t \in [0,T]} \left|\widehat{p}_j(t) - p_j(t) \right| = O( \sqrt{\log n / n} ) \quad \text{ almost surely.}
\]
\end{prop}
We deduce that if the number of discretization points $r_n$ is sufficiently large, that is to say $n^{1/4} = O(r_n)$,   parametric rates for uniform convergence can be obtained for the estimators of $p_1, \ldots,  p_q$,
up to a logarithmic term.

\begin{rem}
\label{remarksplines}
Note that similar uniform consistency results as those stated in Proposition~\ref{propunifpg}, that do not require smoothness conditions on the individual trajectories,  could  also be obtained for estimators based on penalized regression splines. See \cite{Xiao2020} for details on the construction of the estimators of the mean function, as well as asymptotic uniform convergence properties.
\end{rem}

\begin{rem}
The consistent estimation of the covariance function is not directly possible by such smoothing approaches since the required differentiability assumptions are not fulfilled.
 Stochastic processes with values in $\{0, 1\}$ are not mean square differentiable (see the example given in Appendix C in the Supplementary Material). This lack of regularity translates into non differentiability of the covariance function near the diagonal $(t, t)$, as discussed in Chapter 11, \cite{Loeve1978}.
\end{rem}

\subsection{Functional data point of view}

We deduce from \eqref{def:empm}, \eqref{def:empG} and Proposition~\ref{def:trajfda}, the estimators $\widehat{\Gamma}_{j\ell}$ of the cross-covariance operators $\Gamma_{j \ell}$ and  an estimator $\widehat{\boldsymbol{\Gamma}}$ of $\boldsymbol{\Gamma}$. Note that in a more formal way, we can express
\begin{align*}
\widehat{\boldsymbol{p}} & = \frac{1}{n} \sum_{i=1}^n \boldsymbol{X}_i, \\
\widehat{\boldsymbol{\Gamma}} &= \frac{1}{n} \sum_{i=1}^n \boldsymbol{X}_i \otimes \boldsymbol{X}_i - \widehat{\boldsymbol{p}} \otimes \widehat{\boldsymbol{p}} .
\end{align*}

We can state the following consistency and asymptotic normality results which follow immediately from assumption $\mathbf{H}_1$ and Proposition~\ref{def:trajfda}. We denote by $\text{HS}(\mathbb{H})$ the vector space of Hilbert-Schmidt operators on $\mathbb{H}$,  by $\langle.,.\rangle_{\text{HS}}$ the inner product in $\text{HS}(\mathbb{H})$ and $\Gamma \otimes_{\text{HS}} \Delta =  \langle\Gamma,. \rangle_{\text{HS}} \Delta$.

\begin{prop} Suppose that hypothesis $\mathbf{H}_1$ is fulfilled and that for almost all $\omega \in \Omega$, $\mathbf{X}(.,\omega)$ belongs to $D([0,T], \mathbb{R}^q)$, 
then as $n$ tends to infinity, 
\begin{itemize}
\item $\left\| \widehat{\mathbf{p}} - \mathbf{p} \right\|_{\mathbb{H}} \to 0$  almost surely, and 
$\sqrt{n} \left( \widehat{\mathbf{p}} - \mathbf{p} \right) \overset{d}{\longrightarrow} Z$ in $\mathbb{H}$, where $Z$ is a Gaussian random element with mean zero and covariance operator $\boldsymbol{\Gamma}$. 

\item $\left\| \widehat{\boldsymbol{\Gamma}} - \boldsymbol{\Gamma} \right\|_{\text{HS}} \to 0$ almost surely, and  
$\sqrt{n} \left( \widehat{\boldsymbol{\Gamma}} - \boldsymbol{\Gamma} \right) \overset{d}{\longrightarrow} \mathcal{Z}$ in $\text{HS}(\mathbb{H})$,
where $\mathcal{Z}$ is a Gaussian random element with mean zero and covariance operator
\[
\mathbb{E}\left[ \left( (\mathbf{X}-\mathbf{p}) \otimes (\mathbf{X}-\mathbf{p}) -\boldsymbol{\Gamma} \right) \otimes_{\text{HS}} \left( (\mathbf{X}-\mathbf{p}) \otimes (\mathbf{X}-\mathbf{p}) -\boldsymbol{\Gamma} \right) \right].
\]
\end{itemize}
\label{prop:asymp}
\end{prop}

\begin{proof} of Proposition~\ref{prop:asymp} \\ 
{\em 
First note that $\mathbb{H}$ equipped with the inner product $\langle .,. \rangle_{\mathbb{H}}$ is a separable Hilbert space, thanks to Proposition~\ref{def:trajfda}, $\mathbf{X}$ is well defined as a random elements in $\mathbb{H}$.
We clearly have, for each $\omega \in \Omega$,  
\begin{align*}
\| \mathbf{X}(.,\omega) \|_{\mathbb{H}}^2 &=  \sum_{j=1}^q\|X_j(.,\omega)\|^2 \\
& \leq q T , 
\end{align*}
so that all the moments of $\| \mathbf{X} \|_{\mathbb{H}}$ are finite. The proposition is then a direct consequence of  Theorems 8.1.1 and 8.1.2 in \cite{HE2015} that are stated in general separable Hilbert spaces, considering  the empirical mean $\widehat{\mathbf{p}} = \frac{1}{n}\sum_{i=1}^n \boldsymbol{X}_i$ and the empirical covariance operator 
$\widehat{\boldsymbol{\Gamma}} = \frac{1}{n}\sum_{i=1}^n \boldsymbol{X}_i \otimes \boldsymbol{X}_i - \widehat{\mathbf{p}} \otimes \widehat{\mathbf{p}}.$}
\end{proof}

We can deduce from Proposition~\ref{prop:asymp} the consistency and the asymptotic normality of the estimators of the eigenvalues $\widehat{\lambda}_r$, $r=1, \ldots, k$  (see \citealp{DauxoisPR82}  or  Theorem 9.1.3 in \citealp{HE2015}). Under the additional assumption that the eigenvalues are distinct, we can also get the consistency and the asymptotic normality of the estimated eigenfunctions  $\widehat{\boldsymbol{\phi}}_r(t)$, $r=1, \ldots, k$.

\section{An illustration in sensory analysis}
\label{section:application}

The dataset used for illustration deals with sensory experiments. It is well documented in \cite{BNV2023} and can be obtained from a public source. Computations were performed with the library \texttt{MFPCA} (see \citealp{Happ2022}) in the \textbf{R} language  \cite{R2024}. All codes are available on Github, {\small \url{https://github.com/Chemosens/ExternalCode/tree/main/MFPCAWithCategoricalTrajectories}}.

\begin{figure}[!t]
\centering
\includegraphics[scale=.5]{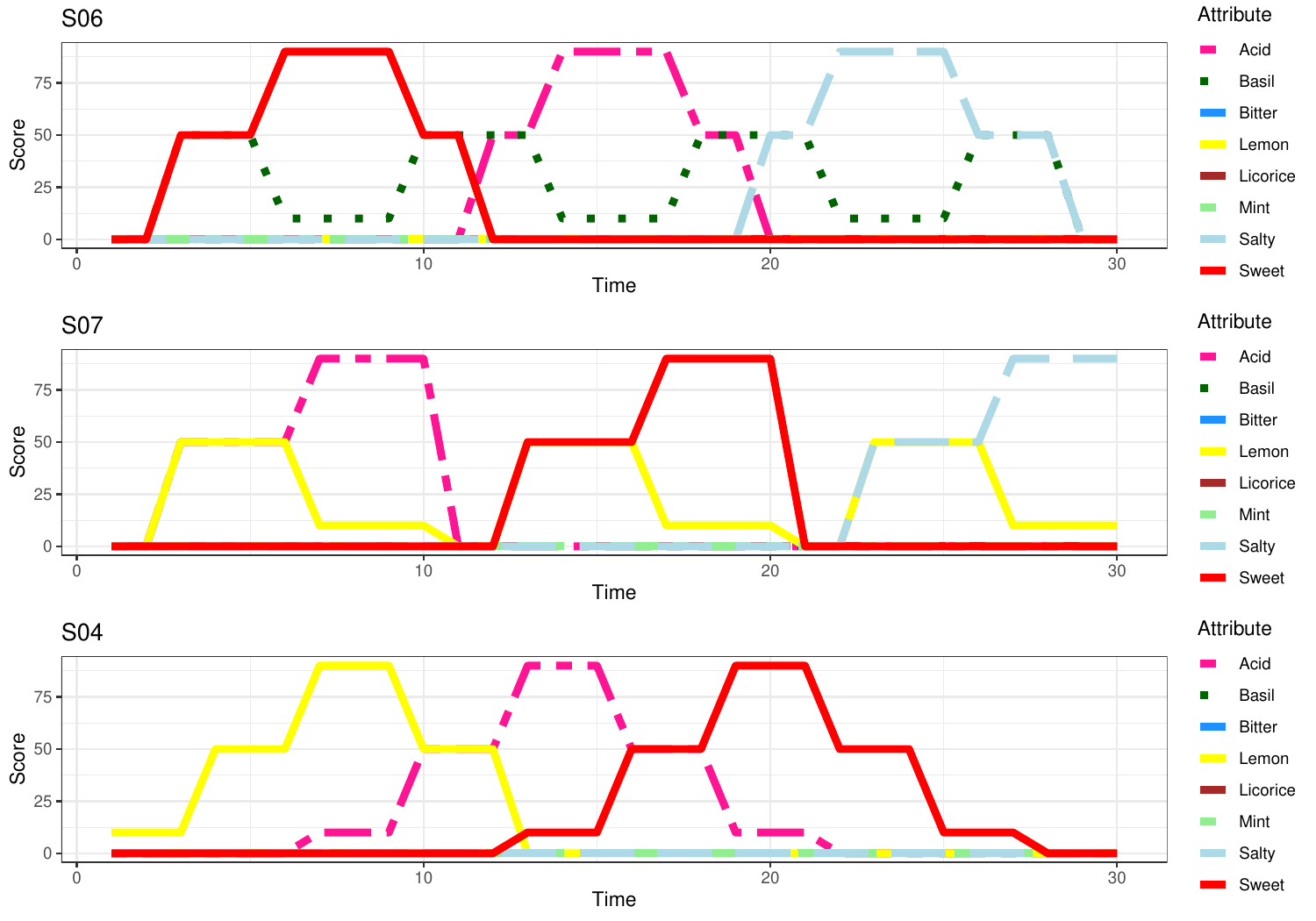}

\caption{Three gustometer-controlled stimuli (S06, S07 and S04) extracted from the open data basis~\cite{BNV2023}.}
\label{fig:Signaux}       
\end{figure}

Fifty participants, called panelists, took part in a tasting experiment and were asked to click on the sensation they perceived in real time from a list of descriptors.
When participants are instructed to click only on the dominant sensation, i.e. when only one perception can be observed at any given time, the protocol is called ``Temporal Dominance of Sensation” (TDS, see \citealp{Pineau_2009} for a reference article). 
When  participants are asked to click on all the sensations they perceive in real time from the same list of descriptors, the protocol is called Temporal Check-All-That-Apply (TCATA, see \citealp{CAS2016} for a seminal article). Unlike the TDS protocol, several descriptors (or none at all) can be selected at any given time for TCATA experiments.
For both cases (TDS and TCATA), the resulting data consists of binary trajectories linked to each state, with the value 0 when a descriptor is unclicked at time $t$ and 1 when it is clicked. The difference between TCATA and TDS data is that at each instant $t$, the sum of all binary trajectories can be different from 1 and take values in the set $\{0,1, \ldots, q\}$ for TCATA whereas the sum is always equal to 1 for TDS. We focus in the following on the TDS analysis (Similar results on TCATA data are reported in the Supplementary Material).

In all the experiments, the tasted temporal solution is  controlled and delivered by a gustometer. Three controlled sensory signals (see Figure~\ref{fig:Signaux}) were tasted by the panelists: S04 (Lemon, followed by Acid, and finally Sweet), S06 (Sweet, Acid, and finally Salty, with a continuous hint of Basil), and S07 (Acid, followed by Sweet, and finally Salty, with a continuous hint of Lemon). The descriptor list included Acid, Sweet, Lemon, Basil, and Salty, alongside distractors such as Mint, Licorice and Bitter, so that the categorical process $Y$ has  $q=8$ states.

The overall number of experiments is $n=150$ and time has been normalized to be $[0,1]$ for each experiment, resulting, for TDS, in the observed categorical trajectories drawn in Figure~\ref{fig:TDS}. 
The average number of jumps per TDS evaluation is 4, with a minimum of 1 and a maximum of 9.
 
 This dataset is particularly relevant because it relies on controlled data, enabling the assessment of result consistency. In contrast, in sensory analyses using conventional products, the ground truth is unknown, making the evaluation of the statistical power of the method more difficult. Here, the use of known stimuli is crucial for validating the interpretation of the results against theoretical expectations.
\subsection{Temporal dominance of sensation (TDS) trajectories}
\label{sect:TDS}

\begin{figure}[t]
\centerline{ \includegraphics[scale=.55]{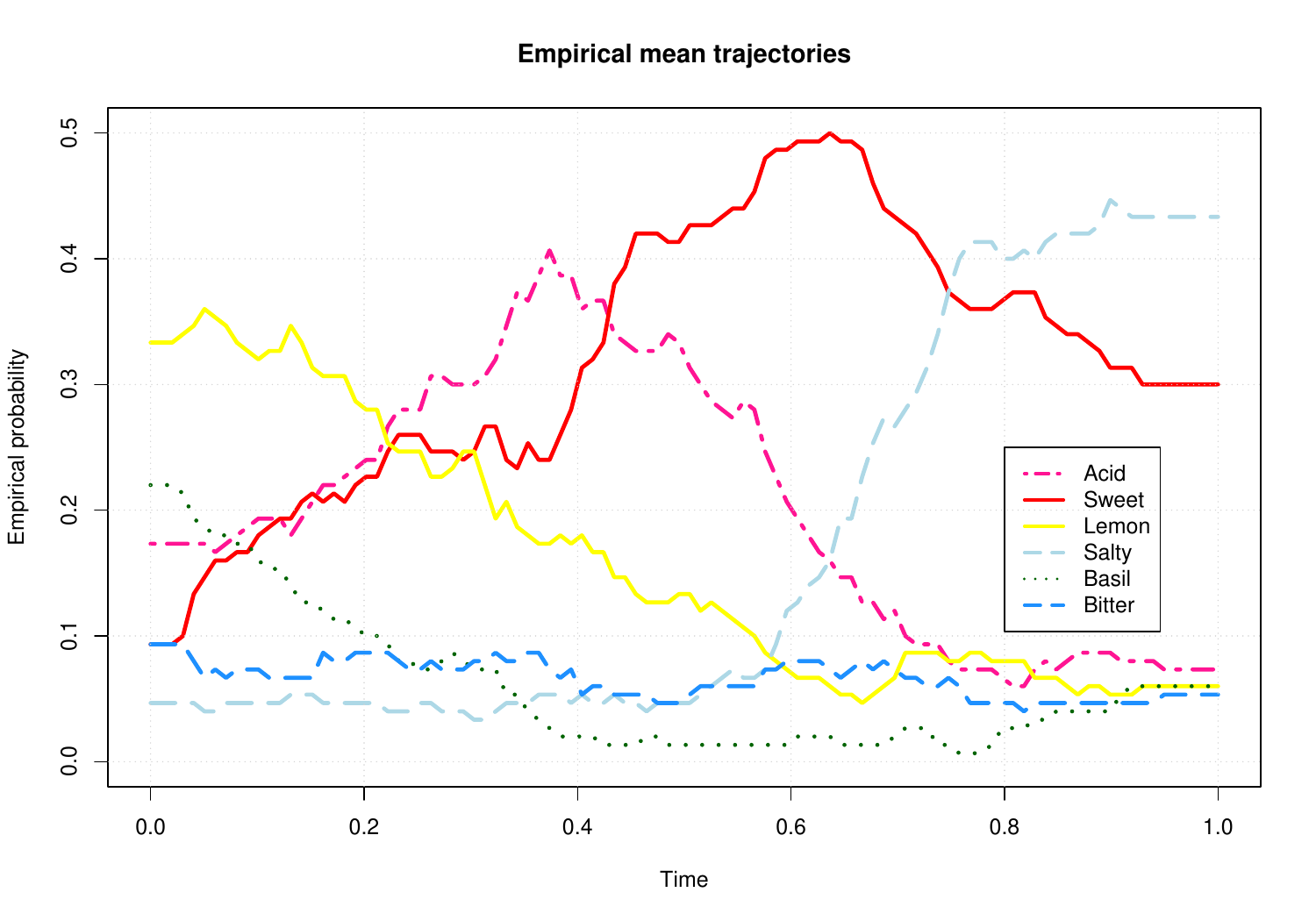}}
%
%
\caption{Empirical probabilities $\widehat{p}_j(t)$, $t \in [0,1]$. The curves are drawn only for  the states $j$ whose average probability of occurrence is larger than 5\%, that is to say $\int_0^1 \widehat{p}_j (t) dt \geq 0.05$.}
\label{fig:meantraj}    
\end{figure}

Most statistical analyses for TDS data  are, until now, based on the examination of the evolution over time of the proportion $\widehat{p}_j(t)$ of occurrence of each state $j$ (see \citealp{Pineau_2009}), as shown in Figure~\ref{fig:meantraj}.
The interpretation is then generally based on the succession of the most frequently observed states, and would lead to say, in that experiment, that the first dominant sensation is Lemon, followed by Acid, Sweet and to finish by Salty.
As seen below, the signal is much more complex since we are, in this controlled experiment, in presence of three different subpopulations (S04, S06 and S07), which are not detected with the average approach that does not take account of individual temporal variations.

We present in the following Sections analyses based on the MPFCA approach introduced in this paper as well as the results of  the continuous time correspondence analysis proposed by \cite{DEV1982}.

\subsection{Analysis based on multivariate functional principal components}

A multivariate FPCA of the trajectories $\mathbf{X}_i$, $i=1, \ldots, n$,  considering equal weights $w_j = 1/q$ for $j=1, \ldots, q$ has been performed, giving the same weight to all  states $S_1, \ldots, S_q$ with a discretization of 100 points and a basis of 
eight B-splines, of order 3, with equispaced interior knots (see the Supplementary Material for the analysis considering the unequal weights scheme given in \eqref{weightstoone}). 
The first eigenvalue $\lambda_1$ represents  23\% of the total variance, whereas the second one  $\lambda_2$ represents 11\% of the total variance, the third one  7\% and the fourth one 6\%. The decrease of the eigenvalues to zero is rather slow (see Figure~\ref{valpMFPCA}), which is not so surprising  since the $\mathbf{X}_i$ trajectories are not continuous.

\begin{figure}[t]
\centerline{
\includegraphics[scale=.55]{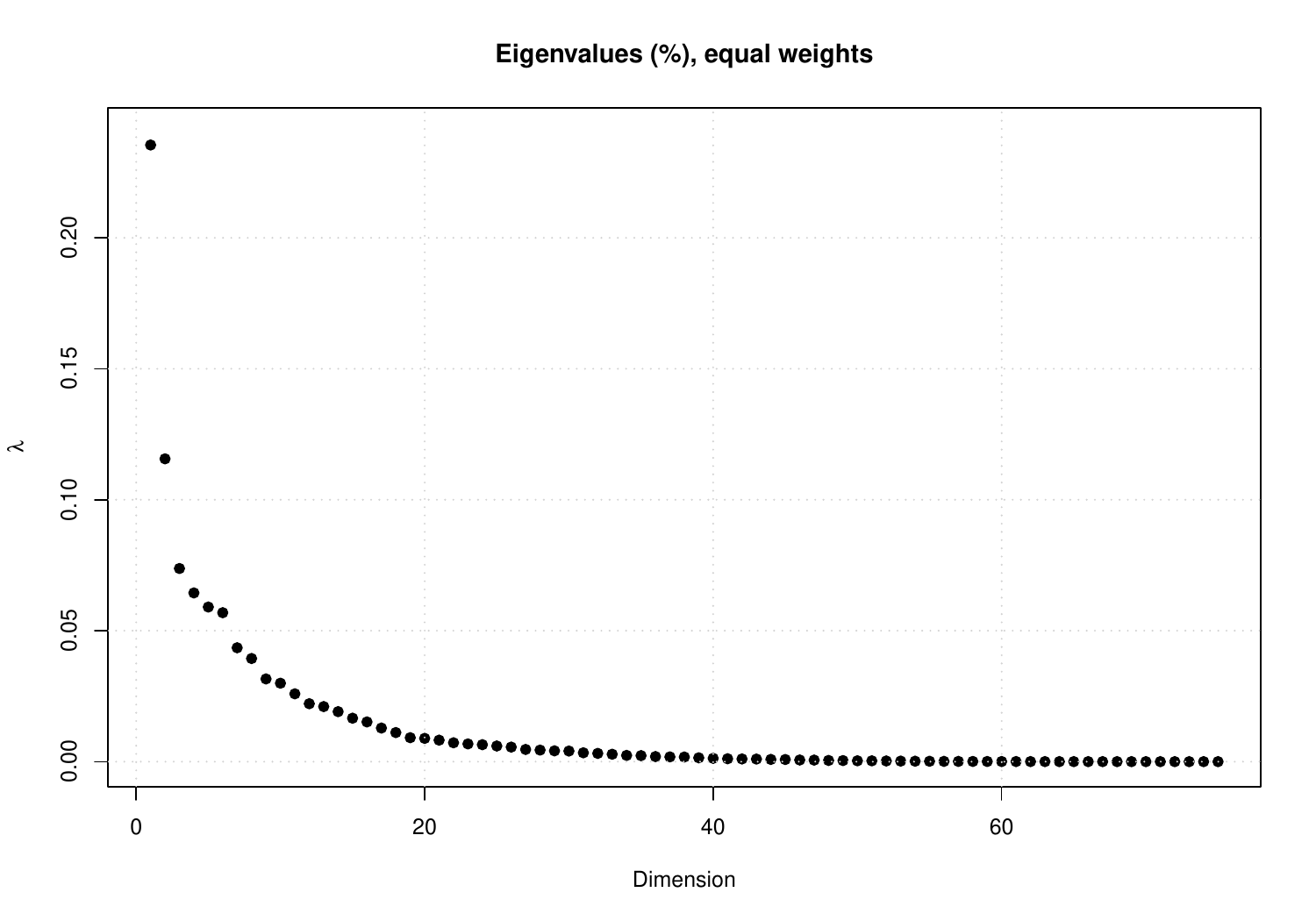}}
\caption{Proportion of total variance captured by the principal components in $\mathbb{H}$.}
\label{valpMFPCA}     
\end{figure}

\begin{figure}[t]
\centering
\includegraphics[scale=.5]{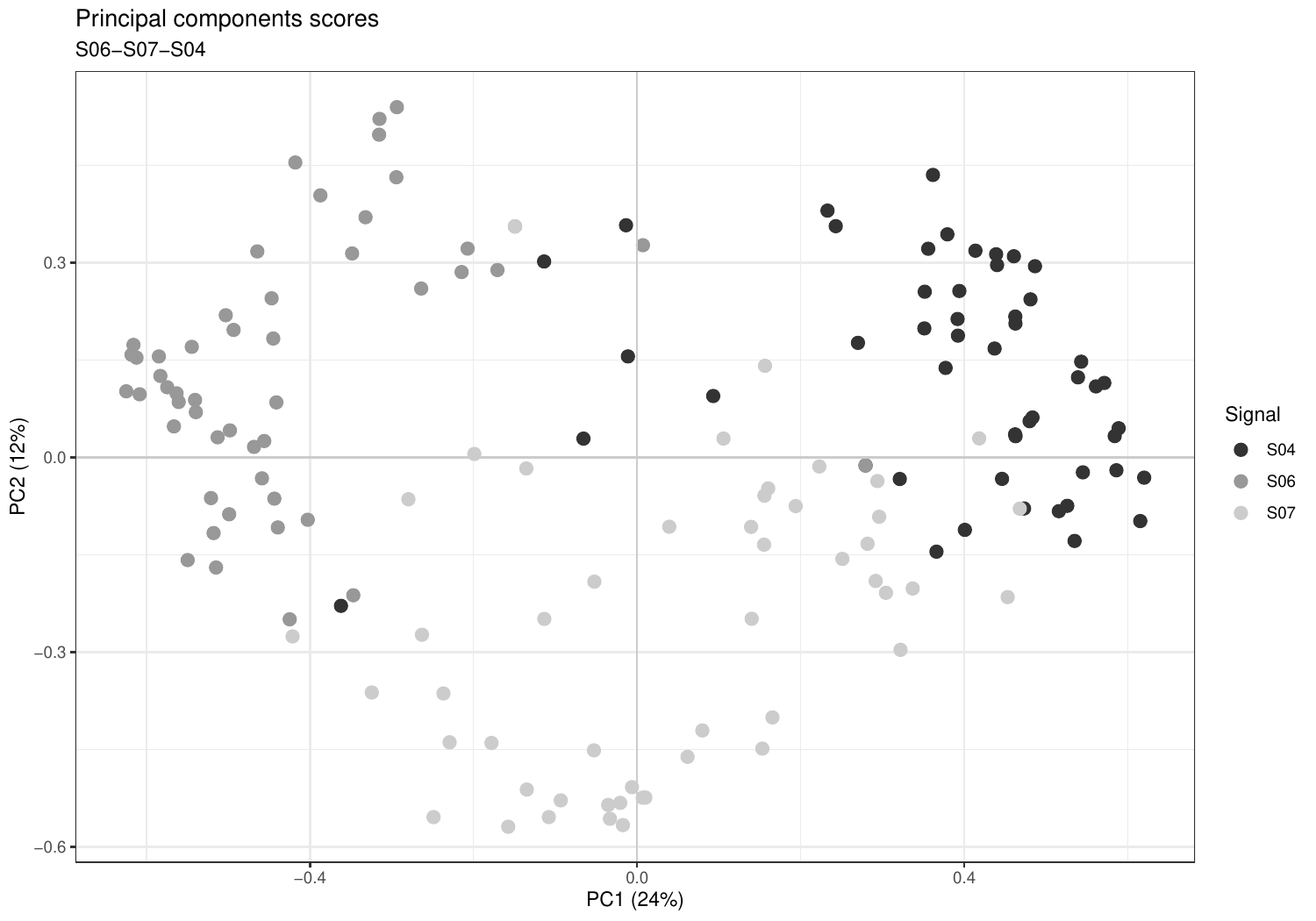}
\caption{Estimated principal component scores. Different gray levels are used to distinguish the observations according to the set of gustomer-controlled stimuli.}
\label{fig:pc}      
\end{figure}

The estimated principal component scores $\langle \mathbf{X}_i - \widehat{\mathbf{p}}, \widehat{\boldsymbol{\phi}}_r \rangle_{\mathbb{H}}$, for $r=1, 2$ are drawn in Figure~\ref{fig:pc}, whereas the main variations around the mean functions, for the first two dimensions, are drawn in Figure~\ref{fig:vp1} and Figure~\ref{fig:vp2}. 
For better interpretation and graphical representation, we only consider the components $j$ with the largest variations, that is to say with the largest values of $\| \widehat{\phi}_{rj} \|$, $j =1, \ldots, q$. Since, for each value of $r$, $\|\widehat{\boldsymbol{\phi}}_r\|_{\mathbb{H}}^2 =1$, we can build the following indicator of importance of each category $j$ in dimension $r$,
\begin{align}
\mbox{imp}_{rj} &= w_j \|\widehat{\phi}_{rj} \|^2, 
\label{def:varimp}
\end{align}
with $\sum_{j=1}^q \mbox{imp}_{rj} = 1$, and consider only the most important variables (see Table~\ref{tab:var_imp}). This leads us to select for graphical representation of the eigenfunctions the states Acid, Lemon, Salty and Sweet for the first dimension and Acid, Basil, Salty and Sweet for the second dimension.

\begin{table}[!t]
\centering
\begin{tabular}{rrrr}
  \hline
 & dim 1 & dim 2 & dim 3 \\ 
  \hline
Acid & 0.08 & 0.26 & 0.42  \\ 
  Basil & 0.04 & 0.07 & 0.00 \\ 
  Bitter & 0.00 & 0.02 & 0.01  \\ 
  Lemon & 0.10 & 0.02 & 0.48 \\ 
  Licorice & 0.00 & 0.00 & 0.00 \\ 
  Mint & 0.00 & 0.00 & 0.00  \\ 
  Salty & 0.22 & 0.30 & 0.02  \\ 
  Sweet & 0.56 & 0.34 & 0.06  \\ 
   \hline
\end{tabular}
\caption{Importance of the different states on each dimension of the MFPCA.}
\label{tab:var_imp}
\end{table}

The results are simple to interpret. For example, the black dots  in Figure~\ref{fig:pc}, corresponding to the S04 experiment, are characterized by a first principal component taking positive values, related, as seen in Figure~\ref{fig:vp1}, to a high  probability of occurrence of Lemon and small probability of occurrence of Sweet at the beginning of the period, and a high  probability of occurrence of Sweet and a small probability of occurrence of Salty at the end of the period. 
These results on Lemon and Sweet are in agreement with the original signal S04, that was first Lemon, then Acid and finally Sweet (see Figure~\ref{fig:Signaux}). It  was the only controlled signal to be not salty at the end of the tasting. Regarding the Acid attribute, it is also present in S06 (middle of tasting) and S07 (beginning of tasting) and this is why it is not highlighted in the first principal components. 

This is the opposite situation for the S06 experiment whose observations (in grey dots) are related to a negative value of the first principal component. The light grey dots, that correspond to the S07 experiment are characterized by negative values of the second principal components. A look at Figure~\ref{fig:vp2} indicates that negative values on the second axis correspond to high probability to be Acid at the beginning of the tasting, Sweet in the middle and Salty at the end. These are exactly the main features of the S07 signal (see Figure~\ref{fig:Signaux}). 

\begin{figure}[t]
\includegraphics[scale=.65]{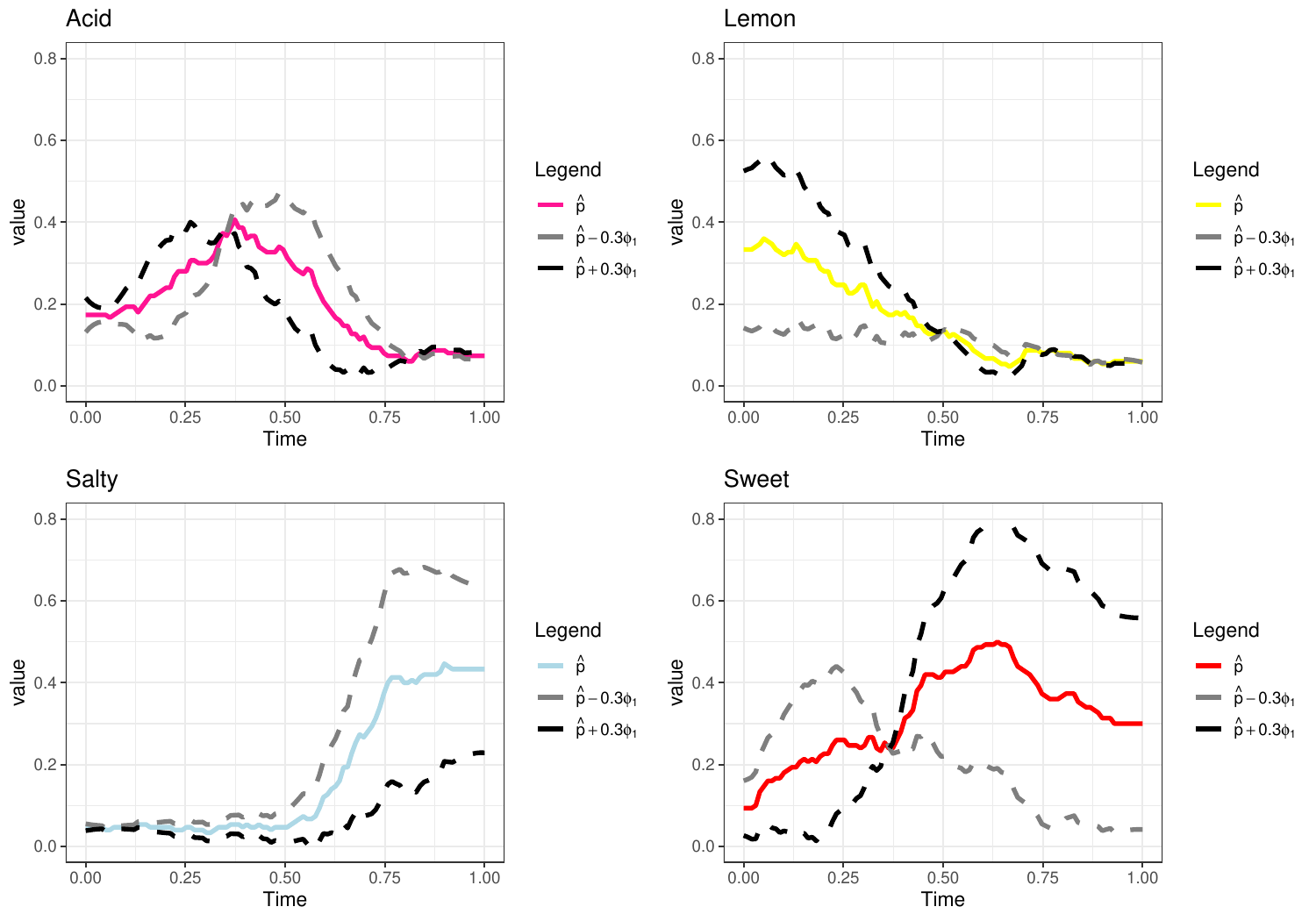}
\caption{MFPCA. First component. Variations around the  probability of occurrence related to the first component of the Karhunen-Loève expansion,  for the most important categories.}
\label{fig:vp1}      
\end{figure}

\begin{figure}
\centering
\includegraphics[scale=.65]{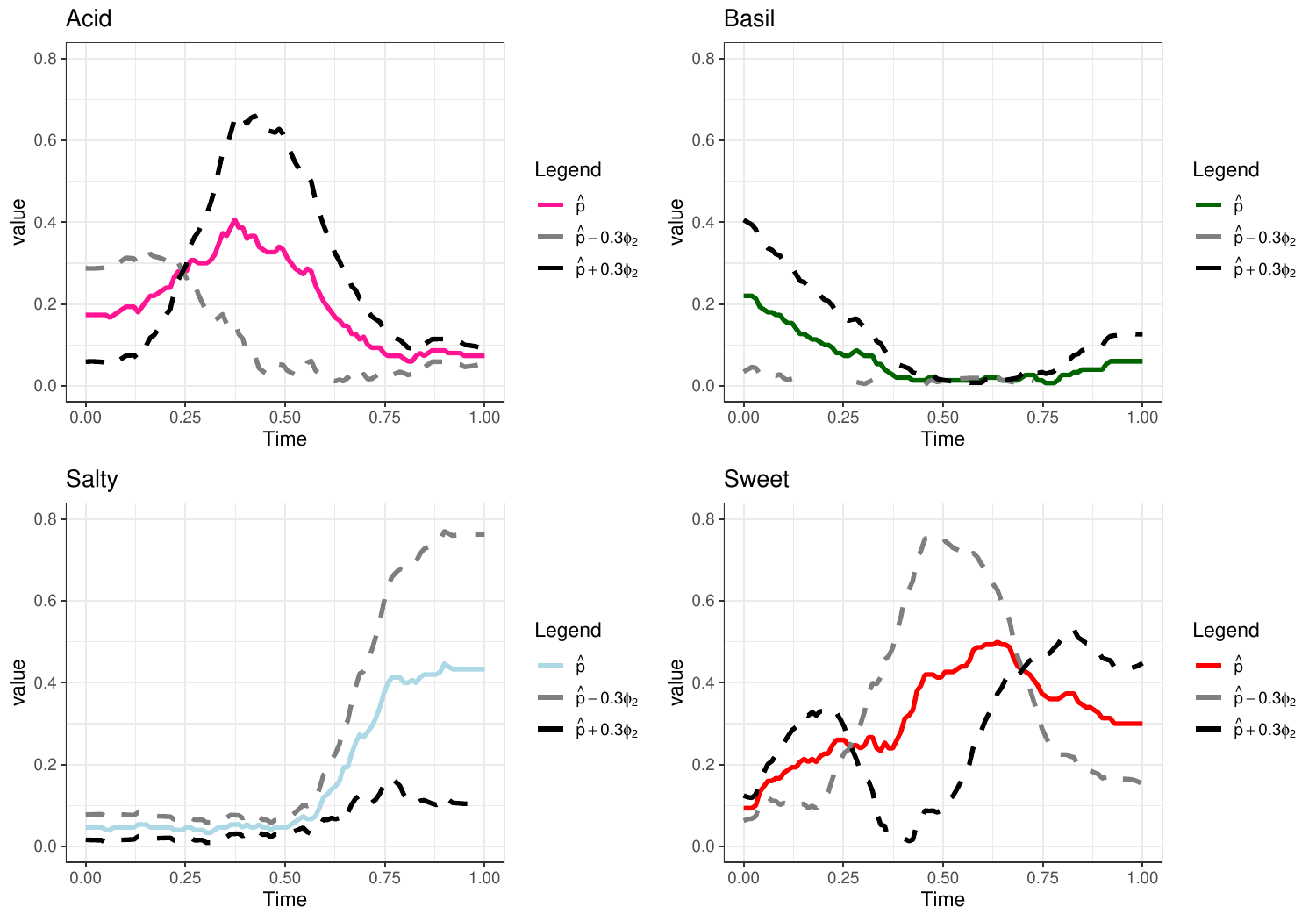}
\caption{MFPCA. Second component. Variations around the  probability of occurrence related to the second component of the Karhunen-Loève expansion,  for the most important categories.}
\label{fig:vp2}      
\end{figure}

\subsection{Analysis based on continuous time correspondence analysis}

An alternative approach in the literature for analyzing TDS data is Categorical Functional Data Analysis (CFDA) described in Remark~\ref{rem:rk1} and which is based on a generalization to continuous time of correspondence analysis. In this study, we employed the  R package cfda \citep{PGV2021}, considering a basis of eight B-splines (of order 3, with equispaced interior knots) to expand the encoding functions.

One limitation of this approach is that it cannot handle events with a zero probability of occurrence, as illustrated in \eqref{cfda:kl}. In practice, this means that it is impossible to compute the encoding function for a specific state at time instants when that state is never selected, an issue observed in Figure~\ref{fig:cfda_harm1} and Figure~\ref{fig:cfda_harm2} for the state Mint and $t \in [0.5, 0.9]$.
Moreover, in our opinion, interpretation is more complicated due to the multiplicative decomposition of the individual trajectories, as seen in \eqref{cfda:kl}. In contrast, deviations from independence in MFPCA rely on an additive decomposition, as seen is \eqref{KLgammajl}, which allows for a more direct interpretation in terms of probability of occurrence.

\begin{figure}
\centering
\includegraphics[scale=.95]{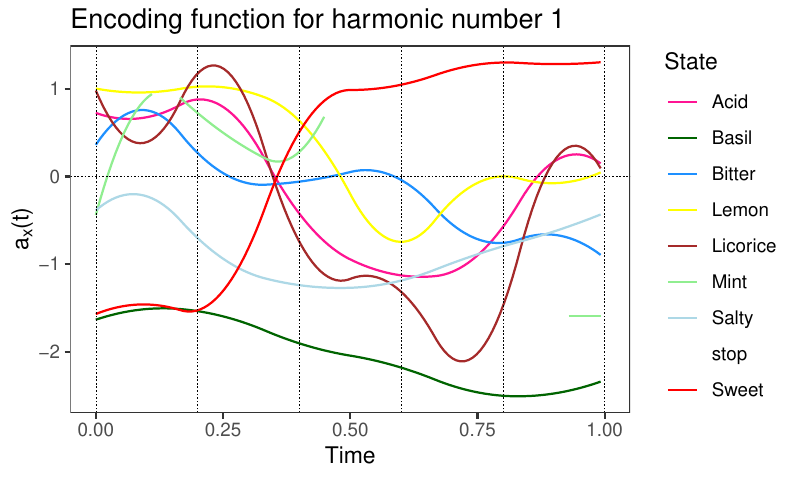}
\caption{Continuous time correspondence analysis. Evolution over time of the first encoding functions (see equation \eqref{def:ecodeDeville}). Note that these functions are not defined when the proportion of occurrence is equal to zero, as seen for the state Mint when $t \in [0.5, 0.9]$.}
\label{fig:cfda_harm1}
\end{figure}

\begin{figure}
\centering
\includegraphics[scale=.95]{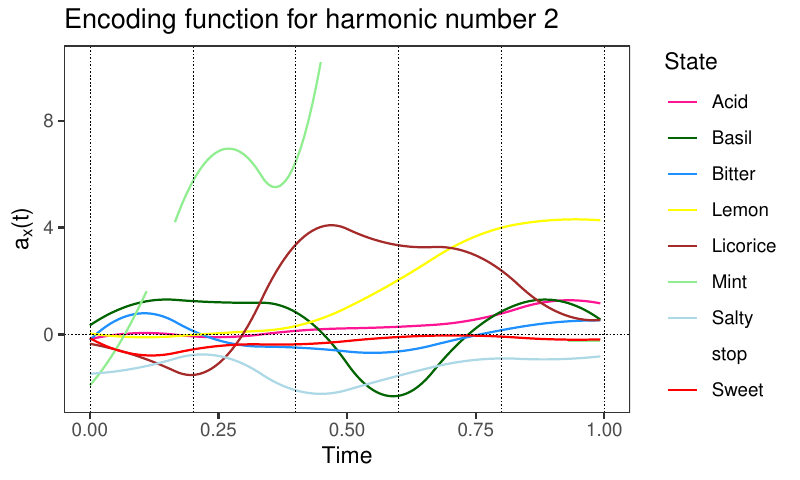}
\caption{Continuous time correspondence analysis. Evolution over time of the first encoding functions (see equation \eqref{def:ecodeDeville}). Note that these functions are not defined when the proportion of occurrence is equal to zero, as seen for the state Mint when $t \in [0.5, 0.9]$.}
\label{fig:cfda_harm2}
\end{figure}

The individual scores computed by the cfda R package are drawn, for the first two dimensions, in Figure~\ref{fig:cfda_ind}.  It reveals that the first axis clearly distinguishes signal S06 (on the left) from signal S04 (on the right), as it was the case for the first principal component scores of MFPCA (see Figure~\ref{fig:pc}). Indeed,  the linear correlation between these two scores is equal to 0.93.

\begin{figure}
\centering
\includegraphics[scale=.85]{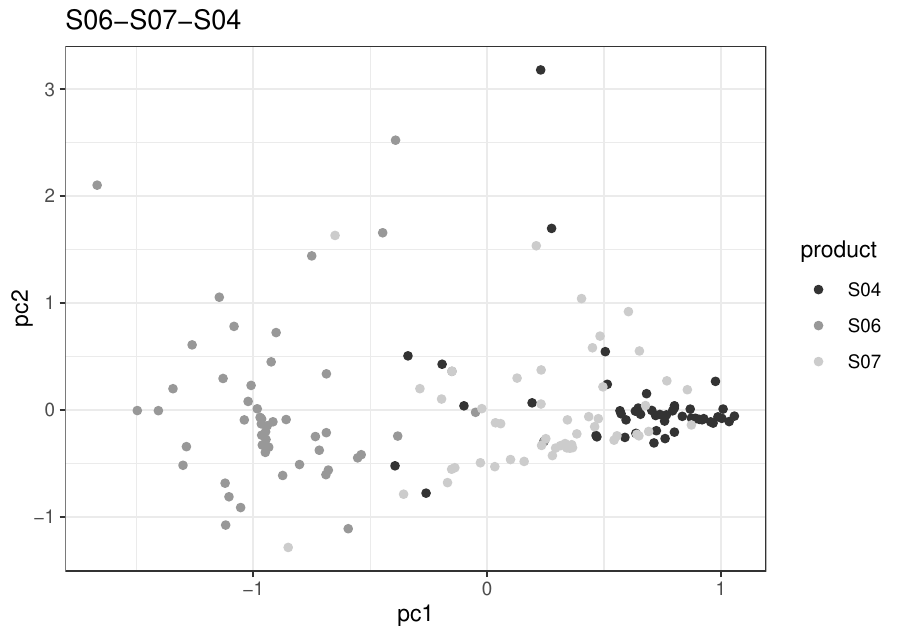}
\caption{CFDA. Individual scores on the first two components.}
\label{fig:cfda_ind} 
\end{figure}

The interpretation of the individual scores based on the encoding functions drawn in Figure~\ref{fig:cfda_harm1} and ~\ref{fig:cfda_harm2} is not so easy and, as it was already done in \cite{DEV1982}, we present on the same map in Figure~\ref{fig:cfda_encoding}, the values of the encoding functions on the first two dimensions, at six different instants in time. 
In this graphical representation, proximity between points indicates a higher probability of co-occurrence, while greater distance suggests a lower probability, as in  correspondence analysis.  
When examining the large values along the first axis, which primarily correspond to trajectories associated with experiment S04, we observe a proximity between Lemon and Licorice at time $t=0$, followed by Acid at $t=0.2$  and Sweet at $t \in \{0.6, 0.8, 1\}$. This pattern suggests that signal S04 is initially perceived as Lemon and Licorice, then transitions to Acid, and is ultimately perceived as Sweet toward the end of the evaluation. This interpretation aligns well with the actual signals, with the exception of Licorice, which served as a distractor and was perceived by only a small fraction of the panelists.
In contrast, signal S06, associated with negative scores on the first axis, is perceived as Sweet at the beginning of the tasting, which is consistent with the real signals.

As illustrated in Figure~\ref{fig:cfda_harm2}, the highest values of the encoding functions on the second dimension are predominantly linked to the Mint descriptor, which is positioned in the upper part of the plot. Since Mint was a distractor in this task, its presence confirms that the trajectories associated with this descriptor can be considered  as outliers. This observation highlights the sensitivity of CFDA to rare events, much like in correspondence analysis.
In the context of CFDA, the second axis appears to primarily capture outlying trajectories, whereas in MFPCA, it distinguishes signal S07 from the other experiments.

\begin{figure}
\centering
\includegraphics[scale=.85]{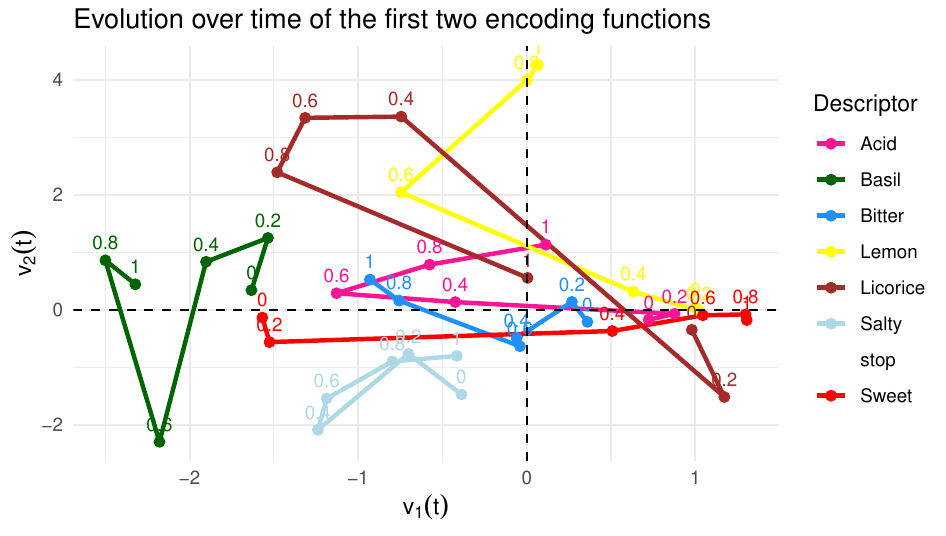}
\caption{Continuous time correspondence analysis. Map of the values of the first two encoding functions for $t \in \{0, 0.2, 0.4, 0.6, 0.8, 1\}$.}
\label{fig:cfda_encoding}
\end{figure}

\section{Concluding remark}

The  way, presented in this work,  of reducing the dimension in a vector space of a panel of categorical trajectories allows for simple interpretation and comparison of individual trajectories.  That statistical methodology can be useful to describe the trajectories related to other experimental protocols such as Temporal Check-All-That-Apply, in which $X_j(t)$ and $X_\ell(t)$, $\ell \neq j$, can be both equal to one at the same time $t$, whereas this would require to increase considerably the number $q$ of states with continuous time correspondence analysis and Markov processes approaches.
It can be easily extended to multivariate categorical trajectories and one could consider simultaneously, in our example,  the three experiments made on the same panelists, that is to say $\mathbf{Y} = (Y_{S04}, Y_{S06}, Y_{S07})$.
Finally, this dimension reduction approach can be useful to build predictive models, such as scalar-on-categorical functional data regression models, permitting to use categorical trajectories as explanatory variables in statistical models via the principal components scores. In the example presented in Section~\ref{section:application}, it is not difficult to predict, with high accuracy, what is the underlying tasting experiment, among S04, S06 and S07, with a simple  linear or quadratic discriminant analysis based on  the principal components scores. 

%
\paragraph{Acknowledgement.}
The authors would like to express their sincere gratitude to the two anonymous reviewers for their valuable comments and suggestions, which greatly contributed to improving the initial version of this manuscript.
The Institut de Mathématiques de Bourgogne (UMR UBE-CNRS 5584) receives support from the EIPHI Graduate School (contract ANR-17-EURE-0002).

\bibliographystyle{apalike}

\clearpage
 
\addcontentsline{toc}{section}{Appendix} \appendix

\begin{center}
\Large \textbf{Statistical description and dimension reduction of categorical trajectories with multivariate functional principal components}
\end{center}

\begin{center}
\Large Supplementary Material
\end{center}

\begin{center}
Hervé \textsc{Cardot} and Caroline \textsc{Peltier} 
\end{center}


\section{Analysis of TDS data considering weighted MFPCA}

We  also consider weights, as those defined in \eqref{weightstoone} in Remark~\ref{rem:2}, and equal to 
\[
w_j = \left( \int_0^1 \widehat{p}_j(t) (1- \widehat{p}_j(t)) dt \right)^{-1}.
\]
Introducing such weights in the inner product in $\mathbb{H}$ gives more importance to the states whose average variance is small, that is to say that are often or very rarely observed (see $w_{p(1-p)}$ in Table~\ref{tab:w} for the numerical values). As noted in  Remark~\ref{rem:2}, considering these weights lead to impose the covariance operators $\Gamma_{jj}$ to have the same trace. As seen in Figure~\ref{valpMFPCAw} on the right, the decrease to zero of the sequence of eigenvalues $\lambda_r$ is slower compared to previous analysis with equal weights. We draw in Figure~\ref{fig:pcW} the first two principal components.

\begin{table}[!h]
\centering
\begin{tabular}{lrrrrrrrr}
  \hline
 & Acid & Basil & Bitter & Lemon & Licorice & Mint & Salty & Sweet \\ 
  \hline
$w_e$ & 0.12 & 0.12 & 0.12 & 0.12 & 0.12 & 0.12 & 0.12 & 0.12 \\ 
  $w_{p(1-p)}$ & 0.02 & 0.06 & 0.05 & 0.02 & 0.21 & 0.60 & 0.03 & 0.01 \\ 
    $w_p$ & 0.02 & 0.05 & 0.05 & 0.02 & 0.21 & 0.62 & 0.02 & 0.01 \\ 
   \hline
\end{tabular}
\caption{Normalized (and rounded at two first digits in the Table) weights used for defining the inner product in $\mathbb{H}$. $w_e$ corresponds to equal weights, $w_{p(1-p)}$ to the scheme given in \eqref{weightstoone} and $w_p$ to weights proportional to the inverse of the average probability of occurrence, $w_j = (\int_0^1 p_j(t) dt)^{-1}$.}
\label{tab:w}
\end{table}

\begin{figure}[t]
\centerline{
\includegraphics[scale=.6]{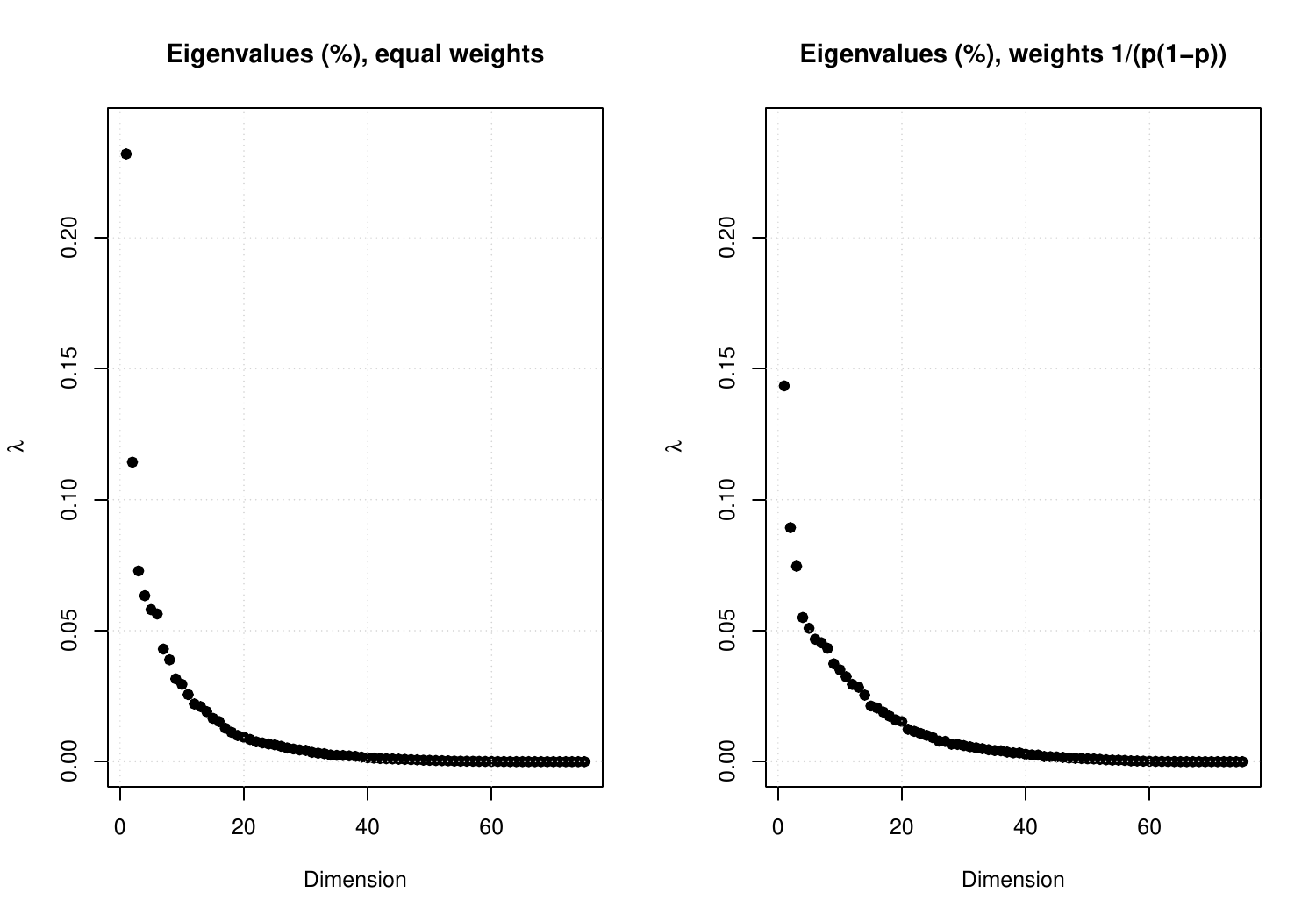}}
\caption{Proportion of total variance captured by the principal components considering equal weights $w_1=\ldots=w_q$ on the left and weights as in  \eqref{weightstoone} on the right.}
\label{valpMFPCAw}     
\end{figure}

\begin{table}[!h]
\centering
\begin{tabular}{rrrr}
  \hline
 & dim 1 & dim 2 & dim 3  \\ 
  \hline
Acid & 0.08 & 0.05 & 0.04 \\ 
  Basil & 0.21 & 0.20 & 0.09  \\ 
  Bitter & 0.00 & 0.02 & 0.14  \\ 
  Lemon & 0.15 & 0.04 & 0.06  \\ 
  Licorice & 0.01 & 0.02 & 0.01  \\ 
  Mint & 0.00 & 0.22 & 0.56  \\ 
  Salty & 0.19 & 0.36 & 0.05  \\ 
  Sweet & 0.36 & 0.08 & 0.05  \\ 
   \hline
\end{tabular}
\caption{Importance of the different states on each dimension of the MFPCA with weights given by \eqref{weightstoone}.}
\label{tab:var_impW}
\end{table}

\begin{figure}[!t]
\centering
\includegraphics[scale=.5]{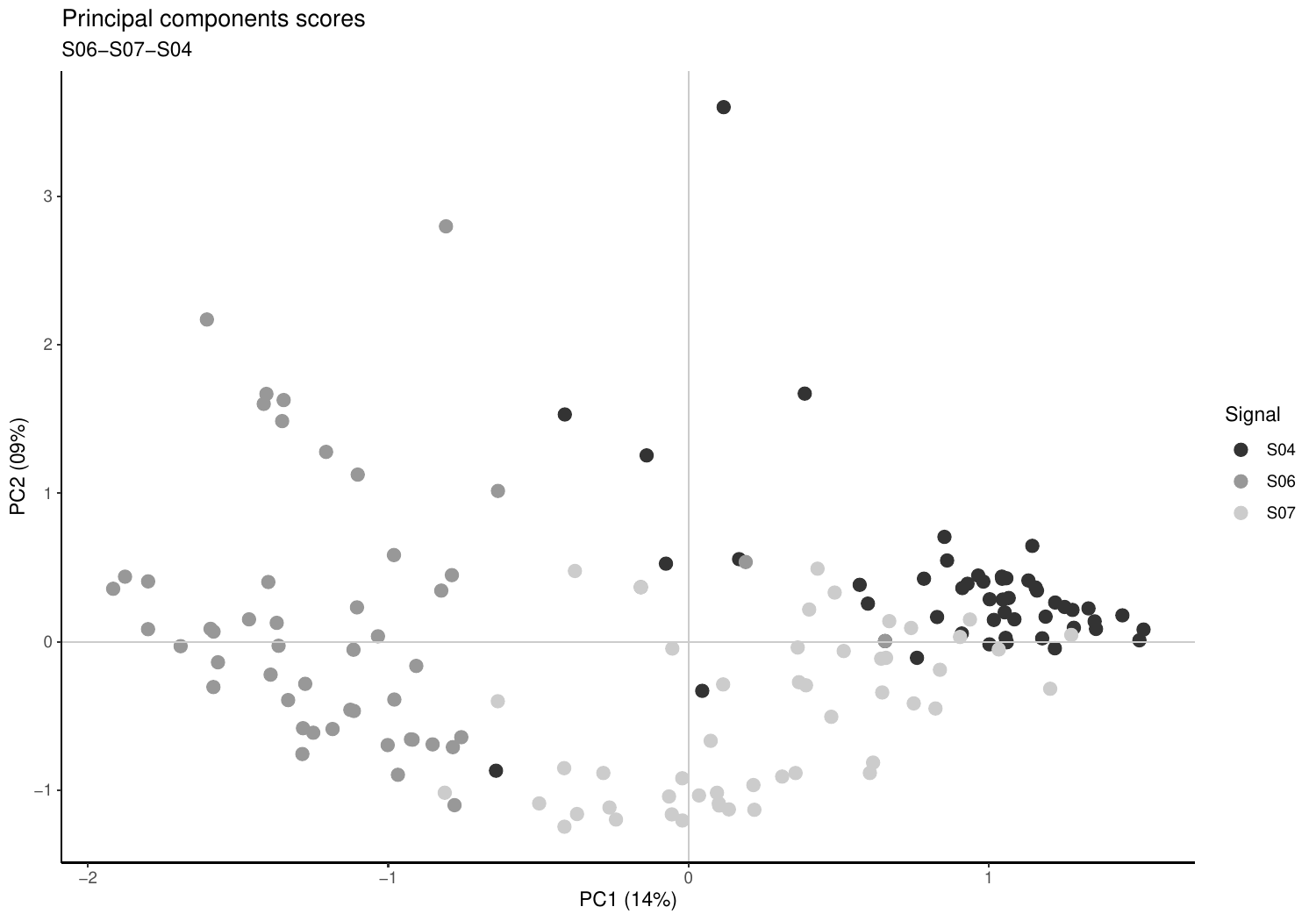}
\caption{Estimated principal component scores with  weights \eqref{weightstoone}. Different colours are used to distinguish the observations according to the set of gustomer-controlled stimuli.}
\label{fig:pcW}      
\end{figure}

The examination of the most important states (see Table~\ref{tab:var_impW}) in the first and second dimension  does  exhibit some little difference with the case of equal weights. First Basil appears to be influential in the first and in the second dimension, amplifying the probability of occurrence for negative value of the first component and positive value of the second component (see Figure~ \ref{fig:vp1W} and Figure~\ref{fig:vp2W}) even if the weight $w_j$ associated to this state is smaller in that unequal weights configuration compared to equal weights MFPCA. This makes it possible to identify, among the S06 tasting experiments, those in which the taste of basil was perceived. 
Another difference is the presence of Mint in the important variable, particularly on the second and third dimension whereas it was not present at all in the equal weights analysis.

\begin{figure}[t]
\includegraphics[scale=.65]{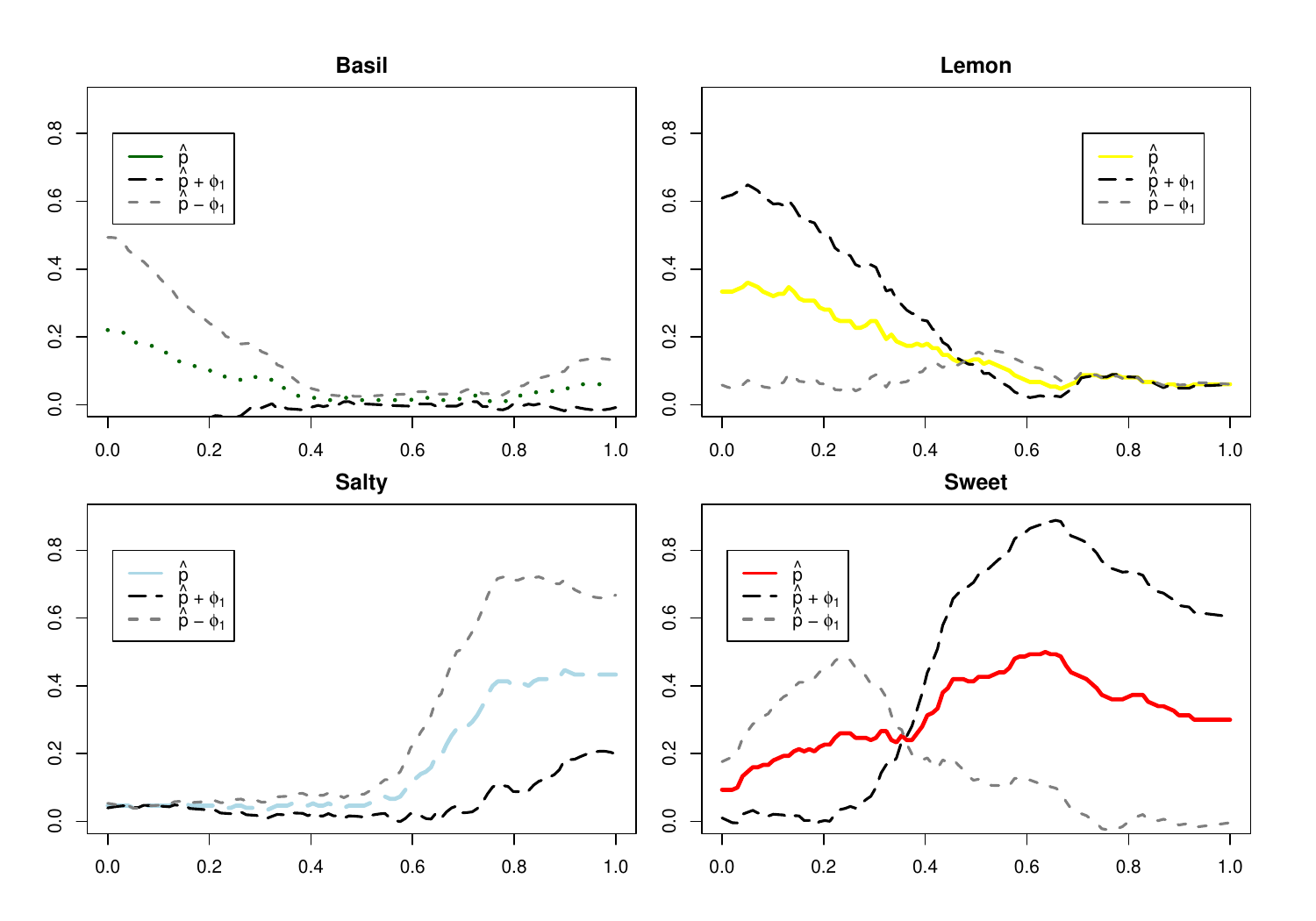}
\caption{MFPCA with unequal weights \eqref{weightstoone}. First component. Variations around the  probability of occurrence related to the first component of the Karhunen-Loève expansion,  for the most important categories.}
\label{fig:vp1W}      
\end{figure}

\begin{figure}
\centering
\includegraphics[scale=.65]{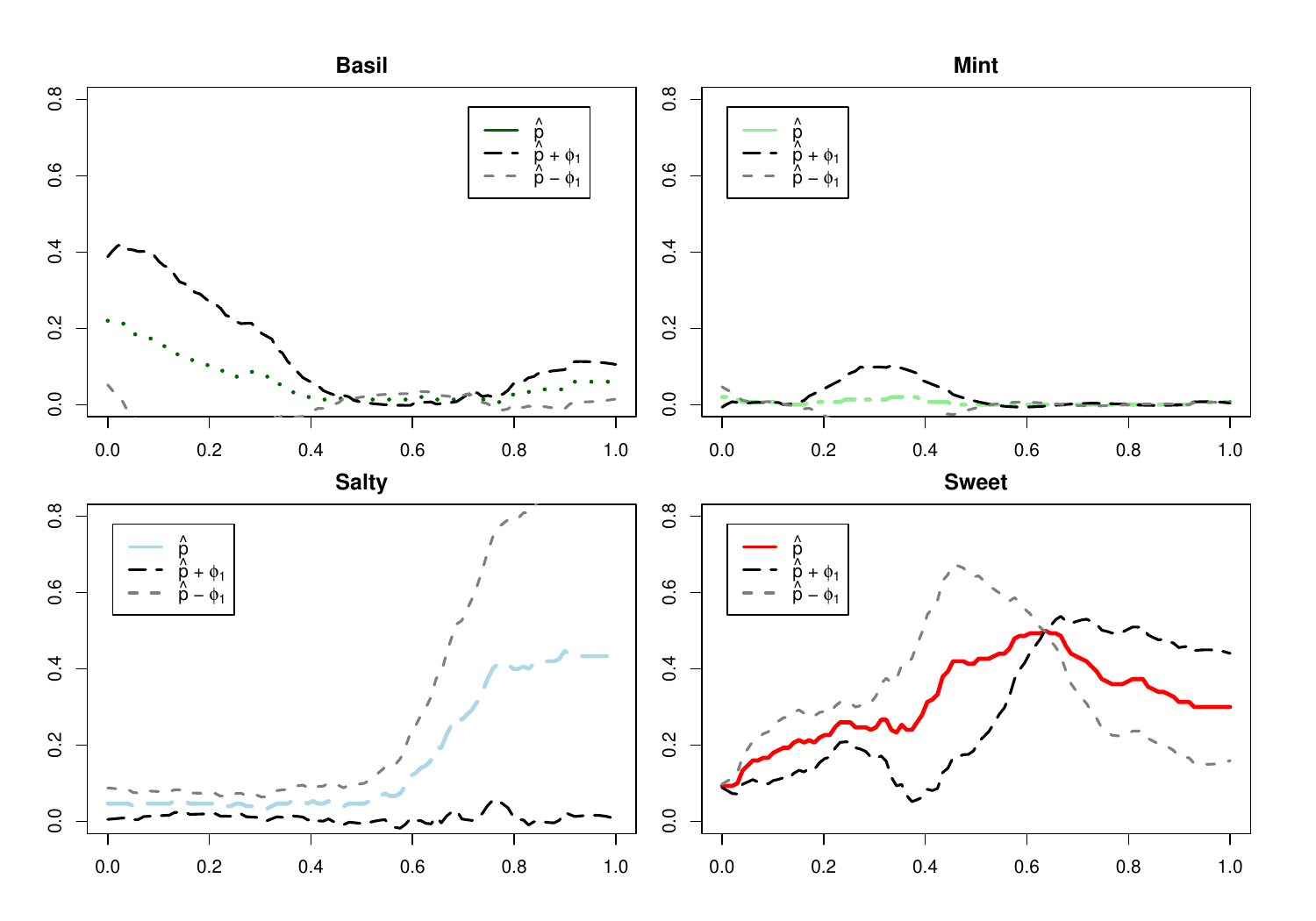}
\caption{MFPCA with unequal weights \eqref{weightstoone}. Second component. Variations around the  probability of occurrence related to the second component of the Karhunen-Loève expansion,  for the most important categories.}
\label{fig:vp2W}      
\end{figure}

\medskip

If we consider weights that take account of the average probability of occurrence, $w_j = \left( \int_0^1 \widehat{p}_j (t) dt \right)^{-1}$,  we remark in Table~\ref{tab:w} that those (normalized) weights $w_p$ 
are very similar to $w_{p(1-p)}$, those given by \eqref{weightstoone}. 
The decomposition of the trajectories are nearly the same and consequently not presented here.

\clearpage

\section{Analysis of TCATA experiments}
\label{section:TCATA}

In the same dataset \cite{BNV2023}, the same gustometer signals (S04, S06, and S07) were also evaluated using the Temporal Check-All-That-Apply protocol but with another panel of fifty panelists.

The mean trajectories $\widehat{p}_j(t)$ for each state $j$ are presented in Figure~\ref{fig:TCATAmean_traj}.  We can remark that their value is equal to 0 in the time interval $[0,0.2]$. This corresponds to a latency time between the start of the tasting and the first click. This latency time is removed in TDS to always have one descriptor selected, but it can be kept here. At the end of the tasting, all the descriptors are automatically unselected, which results in a zero mean value at time $t=1$. The most frequently observed states are Lemon at the beginning of the period, followed by Acid, Sweet, and Salty at the end of the period, as in the TDS evaluations of the signals. As the result of the TCATA experiment can be considered as a set of $q$ binary trajectories, MFPCA can be conducted on them. 

We  have drawn in Figure \ref{fig:TCATmeancitedstates}  the average number of selected states $\sum_{j=1}^q\widehat{p}_j(t)$ at each instant $t \in [0,1]$. Contrary to TDS experiments in which the number of selected states is always equal to one, we note here a variation over time, with a maximum value of 1.5 selected states around time 0.6.
In that sample, the average number of jumps per TCATA trajectory is 7, with a minimum of 1 and a maximum of 21.

The first eigenvalue represents  19\%  of the  total  variance, whereas the second one represents 12 \% of the total variance (see Figure~\ref{fig:TCATAeigenvalues}).  As seen in Table~\ref{tab:var_impTCATA} the importance of the different states is a slightly different compared to the analysis of TDS data. We note that the Basil state  appears to be important in both the first and second dimension, allowing to distinguish the S06 experiment, the only experiment with basil flavours, from the two other tasting experiments.  

\begin{figure}
\centering
\includegraphics[scale=.5]{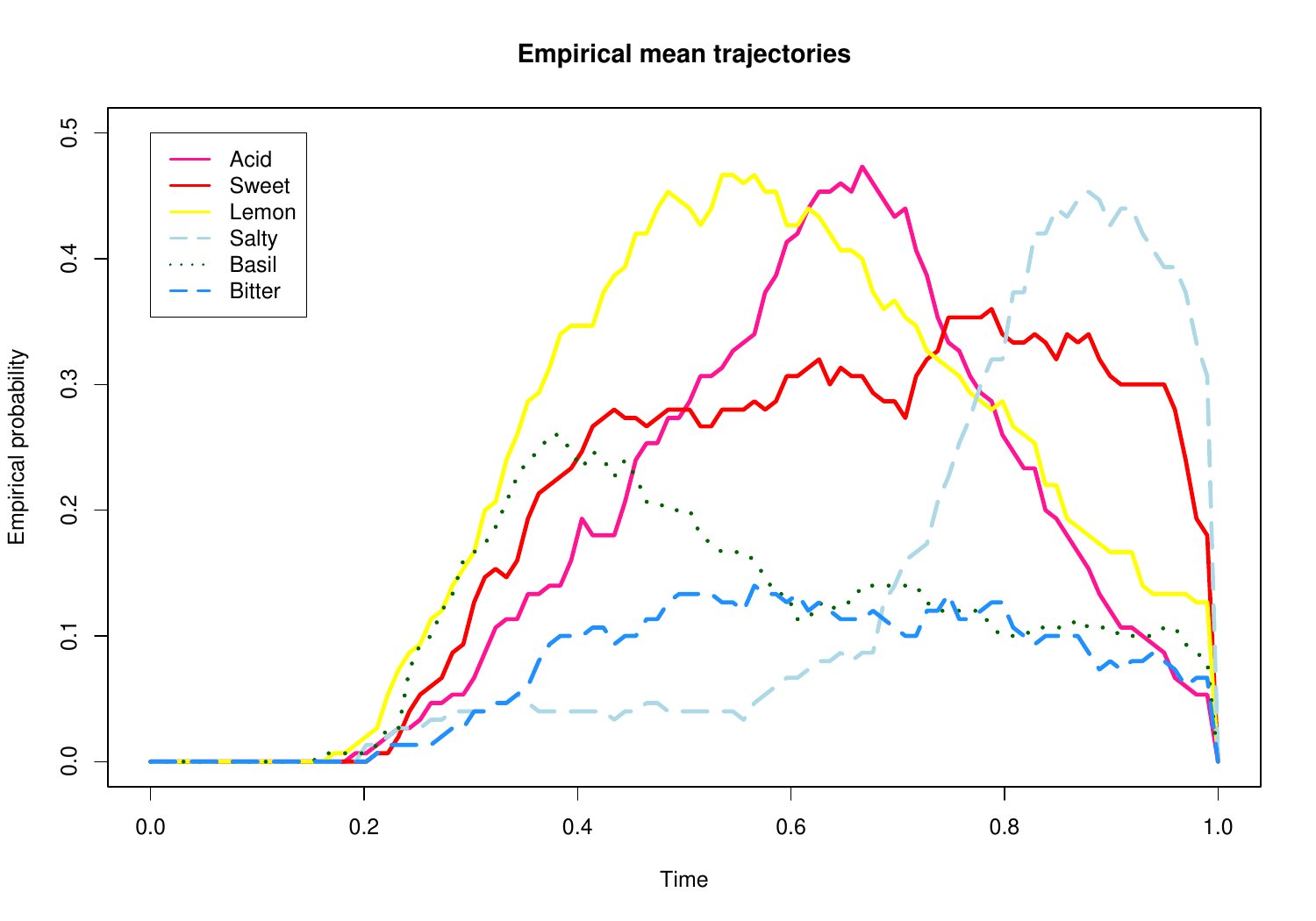}
\caption{TCATA experiment. Empirical probabilities of occurrence $\widehat{p}_j(t)$, for the most frequently selected states.}
\label{fig:TCATAmean_traj}       
\end{figure}

The estimated principal component scores are  drawn  in Figure~\ref{fig:TCATAind}, whereas the main variations around the mean functions, for the first two dimensions, are drawn in Figure~\ref{fig:TCATAvp1} and Figure~\ref{fig:TCATAvp2} for the  same states as those selected in the TDS sample.  
The points corresponding to the S04 experiment are characterized by a first  principal component  taking  positive values, related, as seen in Figure~\ref{fig:TCATAvp1} to a high probability of occurrence of Lemon in the middle of the tasting, a small probability of occurrence of Acid at the end of the period, and a high probability of occurrence of Sweet and a small probability of occurrence of Salty at the end of the period. Nothing appears on the first component for Acid perception. On the second component, the same points have negative scores, showing a high probability of Acid occurring after 0.5. 
The conclusions on the first two axes allow us to recognize the main structure of the original signal (see Figure~\ref{fig:Signaux}). 
The same reasoning can be applied to S06 and S07, which are well discriminated by the first two dimensions of the 1-2 map.

\begin{figure}
\centering
\includegraphics[scale=.5]{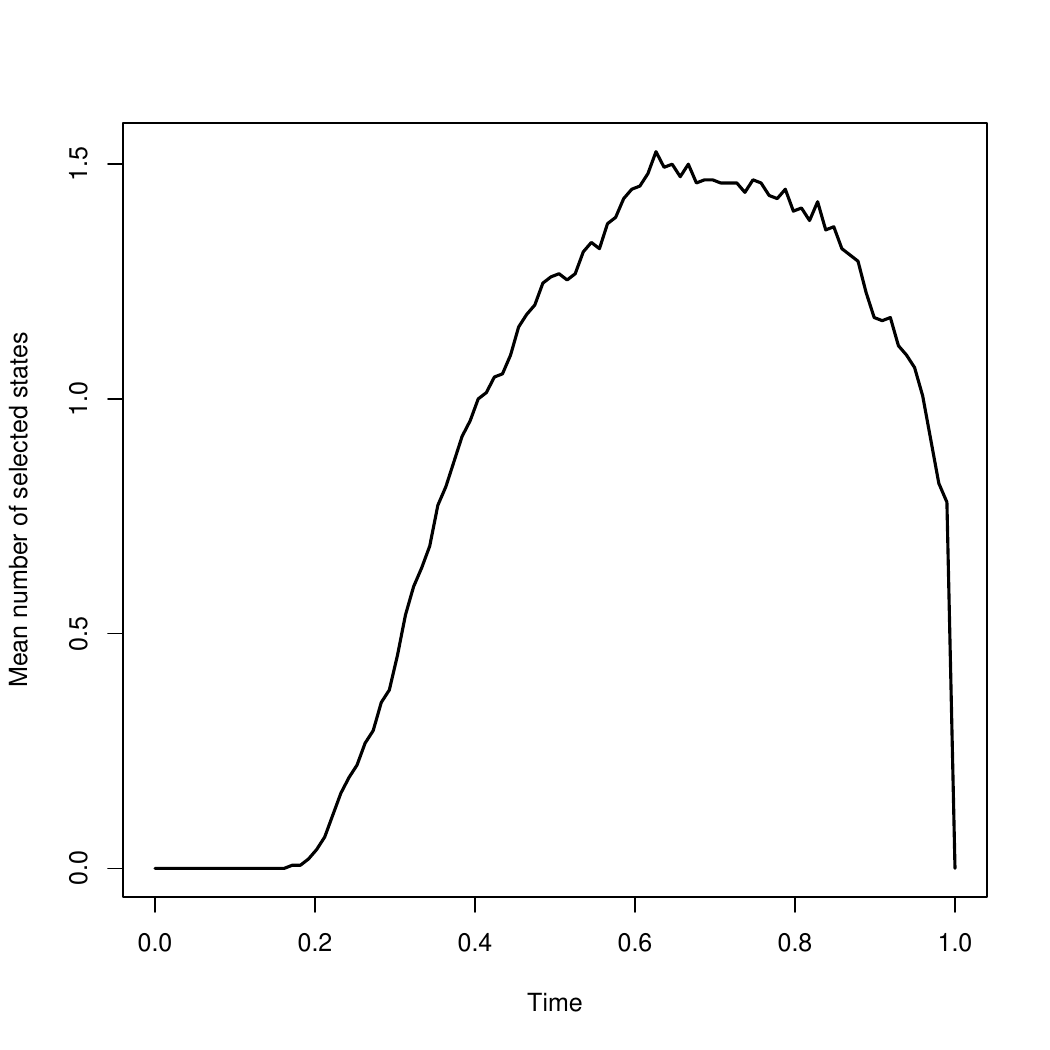}
\caption{TCATA experiment. Mean number $\sum_{j=1}^D\widehat{p}_j(t)$ of selected states over time.}
\label{fig:TCATmeancitedstates}       
\end{figure}

\begin{table}[ht]
\centering
\begin{tabular}{rrrr}
  \hline
& dim 1 & dim 2 & dim 3 \\
  \hline
Acid & 0.00 & 0.35 & 0.20 \\
  Basil & 0.19 & 0.19 & 0.04 \\
  Bitter & 0.01 & 0.00 & 0.01 \\
  Lemon & 0.27 & 0.11 & 0.39 \\
  Licorice & 0.00 & 0.01 & 0.00 \\
  Mint & 0.00 & 0.03 & 0.02 \\
  Salty & 0.16 & 0.11 & 0.20 \\
  Sweet & 0.36 & 0.19 & 0.14 \\
   \hline
\end{tabular}
\caption{Importance of the different states on each dimension of the MFPCA of TCATA experiments, with  equal weights.}
\label{tab:var_impTCATA}
\end{table}

\begin{figure}
\centering
\includegraphics[scale=.5]{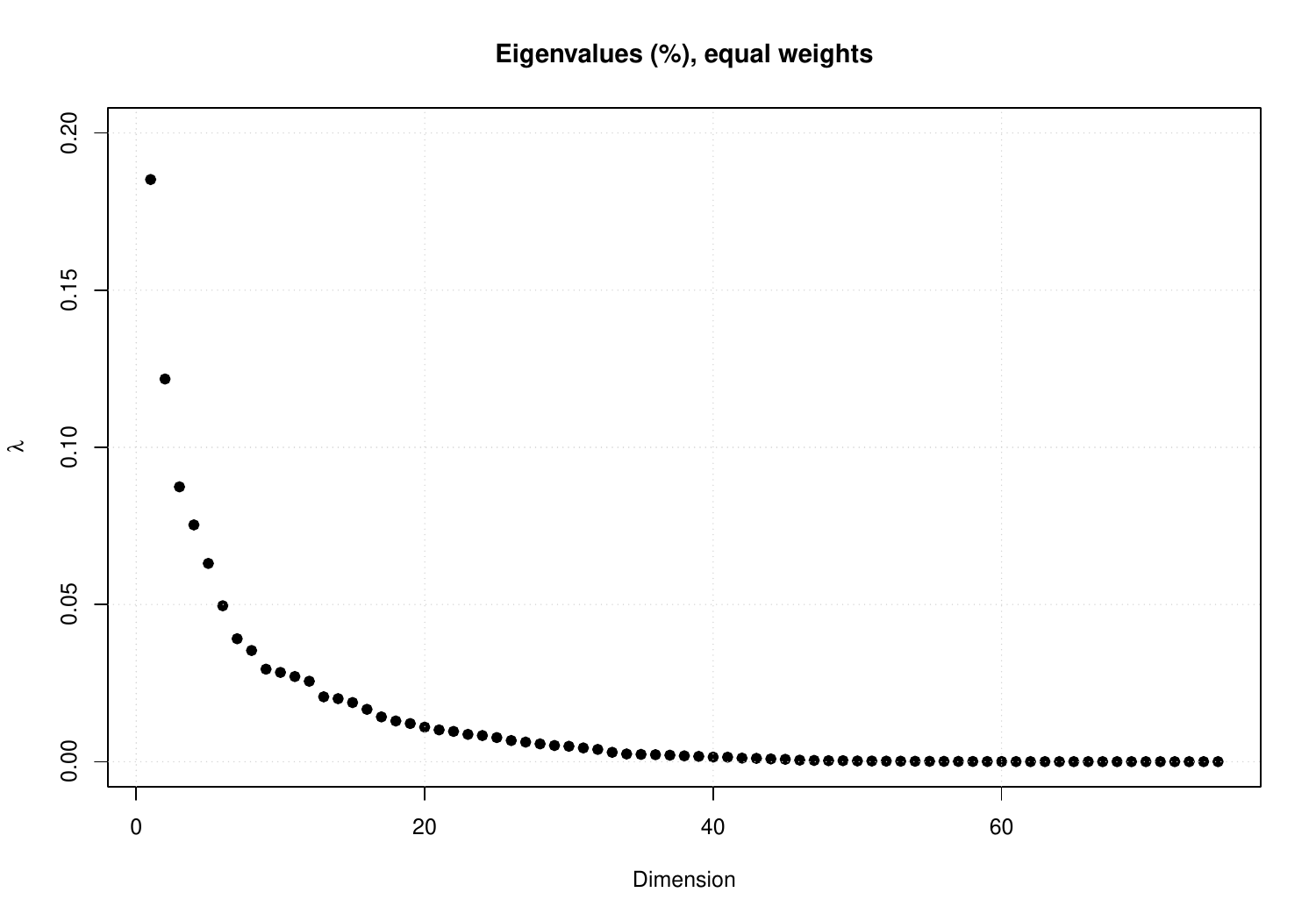}
\caption{Proportion of total variance captured by the principal components considering equal weights for TCATA experiment.}
\label{fig:TCATAeigenvalues}       
\end{figure}

\begin{figure}
\centering
\includegraphics[scale=.5]{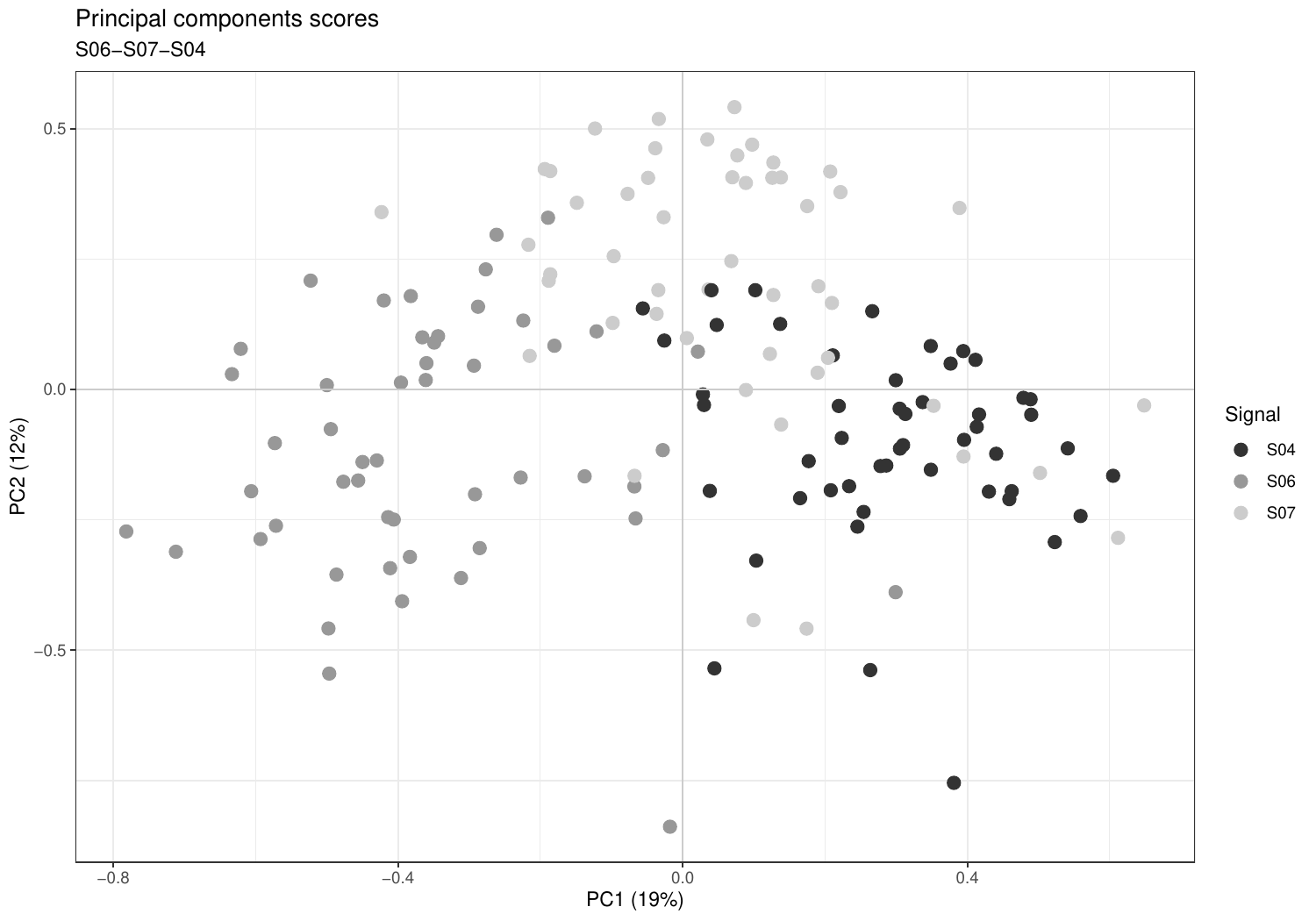}
\caption{Estimated MFCPA scores with TCATA data. Different levels of gray are used to distinguish the observations according to the set of gustometer-controlled stimuli.}
\label{fig:TCATAind}       
\end{figure}

\begin{figure}
\centering
\includegraphics[scale=.65]{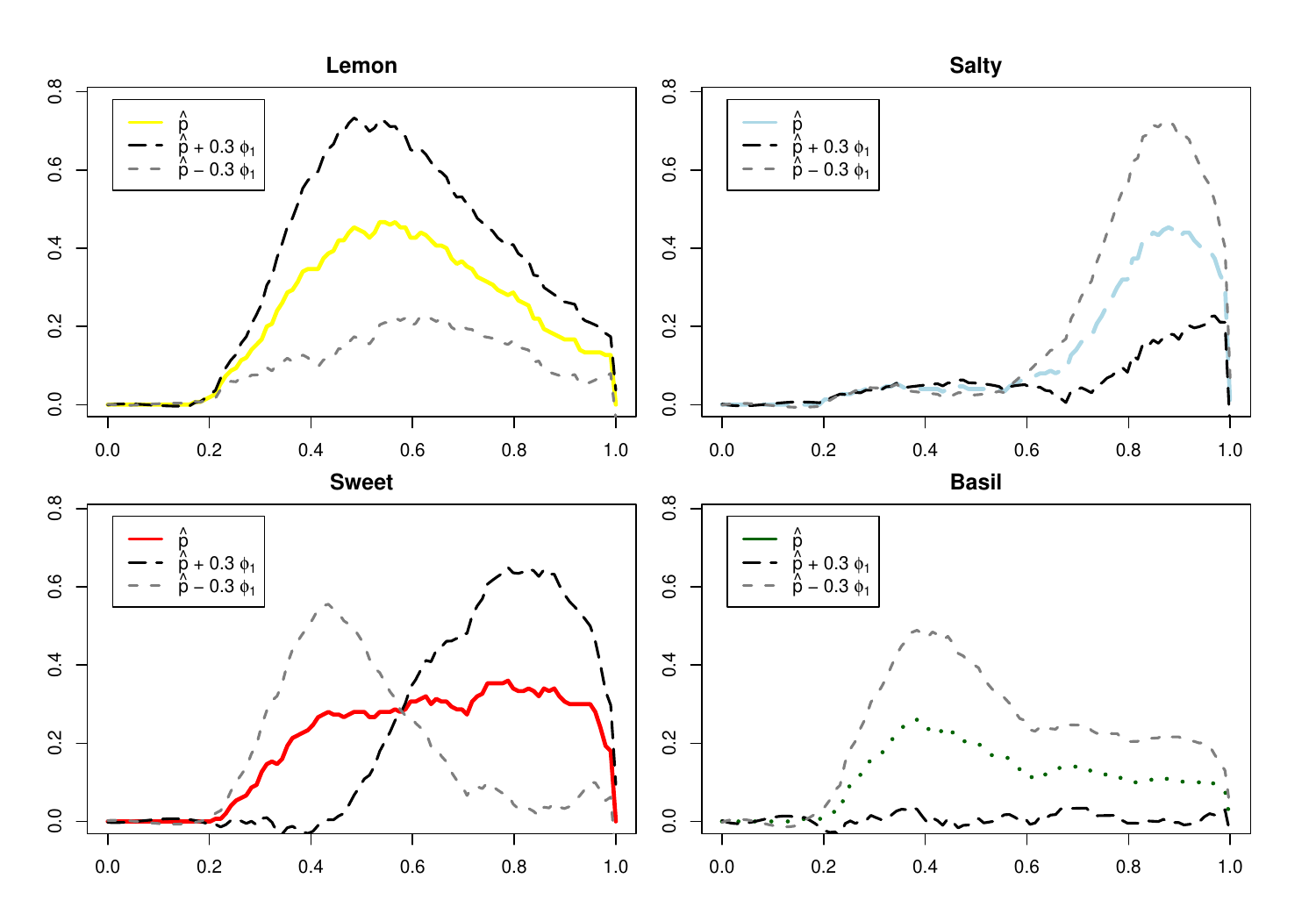}
\caption{MFPCA of TCATA data. First component. Variations around the  probability of occurrence related to the first component of the Karhunen-Loève expansion,  for the most important categories.}
\label{fig:TCATAvp1}       
\end{figure}

\begin{figure}
\centering
\includegraphics[scale=.65]{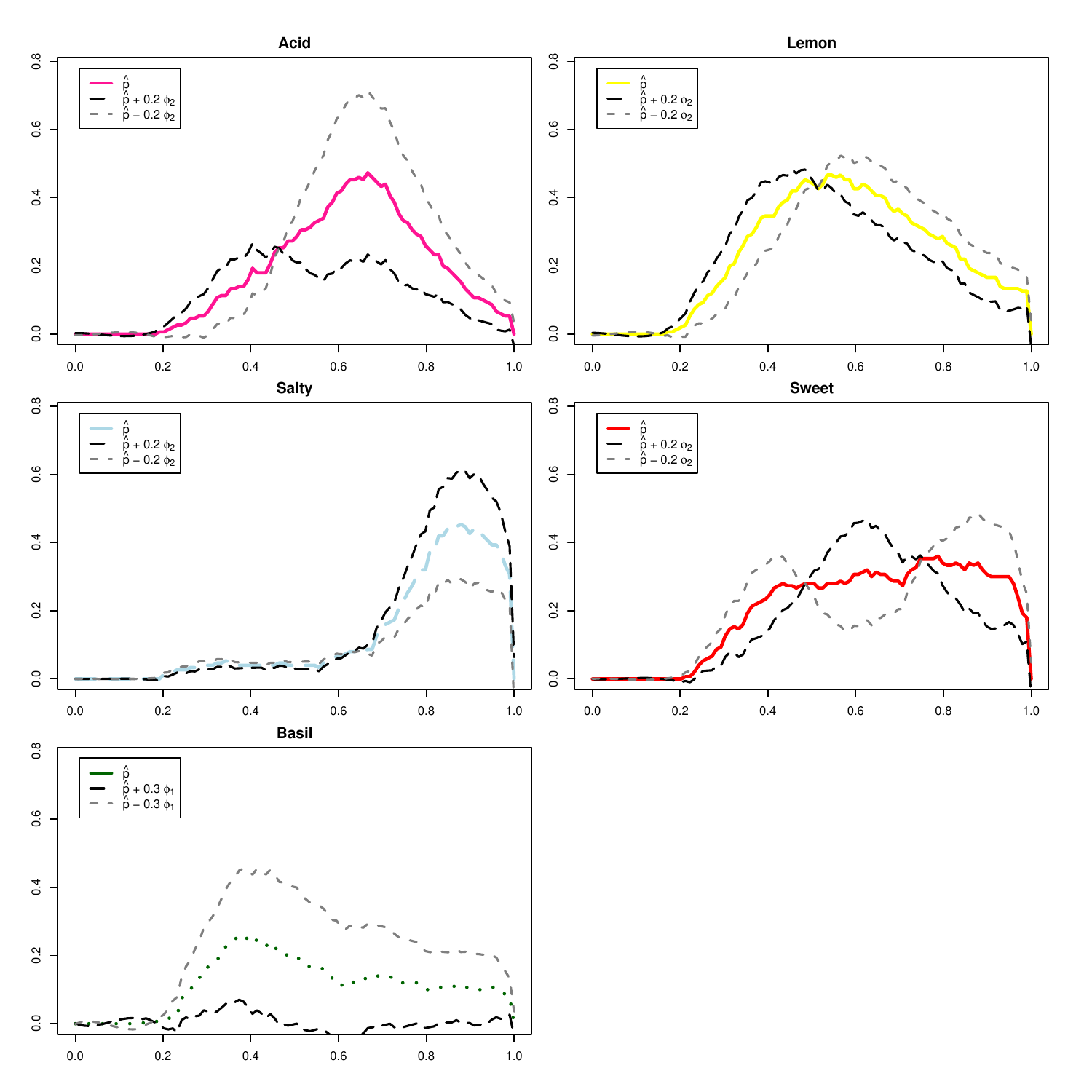}
\caption{MFPCA of TCATA data. Second component. Variations around the  probability of occurrence related to the second component of the Karhunen-Loève expansion, for the most important categories.}
\label{fig:TCATAvp2}       
\end{figure}


\section{A simple example of process taking values in $D([0,T], \mathbb{R})$}

We present in this Section a simple example of stochastic process whose trajectories are not continuous and belong to $D([0,T], \mathbb{R})$. The stochastic process  is continuous in probability, with mean trajectory that is a smooth function and a covariance function that is not differentiable in the neighborhood of the diagonal.
This example is derived from the well known telegraph process. Define 
\begin{align*}
   X(t) = \dfrac{(-1)^{N(t)}+1}{2},
\end{align*}
where   $N(t)$ is a Poisson process  defined on $[0,T]$, with intensity $\lambda >0$. 
For every $t \in [0,T]$, $X(t) \in \{0,1\}$. It is well known that the trajectories of $N(t)$ are not continuous, and the number of jumps (discontinuity points) on $[0,T]$ follows a Poisson distribution with parameter $\lambda T$. Thus, for all $\omega \in \Omega$, the number of jumps of $X(.,\omega)$ is finite and $X(.,\omega) \in D([0,T], \mathbb{R})$. 

Furthermore, the trajectories of $X(t)$  are continuous in probability. Indeed, we have 
\[
X(t+h) - X(t) = \frac{(-1)^{N(t+h)} - (-1)^{N(t)}}{2},
\]
and  $\mathbb{P}[X(t) \neq X(t+h)]$ is equal to the probability that $N(t+h) - N(t)$ is an odd number, which is equal to 
\begin{align*}
   \mathbb{P}[X(t) \neq X(t+h)] &= \sum_{k=0}^{\infty} \mathbb{P}(N(t+h) - N(t) = 2k + 1) \\
   &= \sum_{k=0}^{\infty} \dfrac{e^{-\lambda h}(\lambda h)^{2k +1}}{(2k +1)!} \\
   &= e^{-\lambda h} \sinh(\lambda h). 
\end{align*}

For sufficiently small $h$, we get with a Taylor expansion that $e^{-\lambda h} \sinh(\lambda h) =  \lambda h + o(h)$ and  
$\lim_{h \to 0} \mathbb{P}[X(t) \neq X(t+h)] =0$, showing that for all $t \in [0,T]$, the process $X(t)$ is continuous in probability. Furthermore it is possible to get an explicit expression for its expectation, whose derivatives of any order are well defined. 
First note that 
\begin{align*}
\mathbb{E}[(-1)^{N(t)}] & = \sum_{k=0}^{\infty} (-1)^k \mathbb{P}[N(t)=k] \\
 &= \exp(-2\lambda t),
\end{align*}
so that 
\[
\mathbb{E}[X(t)] =  \frac{\mathbb{E}[(-1)^{N(t)}] + 1}{2} = \frac{e^{-2\lambda t} + 1}{2}.
\]
Using the fact that, for $t \geq s$, $N(t) - N(s)$ has a Poisson distribution, with parameter $\lambda (t-s)$, we get that 
\[
\text{Cov}(X(s), X(t)) = \frac{e^{-2\lambda |t-s|} - e^{-2\lambda (s + t)}}{4}.
\]
The partial derivatives of $\text{Cov}(X(s), X(t))$ are not defined on the diagonal $(t,t)$ and we can deduce  (see \citealp{Loeve1978}, Chapter 11) that the stochastic process $X(t)$ is not mean square differentiable.  
Indeed, we can check that the process $X(t)$ is not mean square differentiable. We have
\begin{align*}
  \mathbb{E}\left[\left(\frac{X(t+h) - X(t)}{h}\right)^2\right] 
  & = \frac{1}{4h^2} \mathbb{E}\left[ \left( (-1)^{N(t+h)} - (-1)^{N(t)} \right)^2 \right] \\
  &= \frac{e^{-\lambda h} \sinh(\lambda h)}{4h^2}  .
\end{align*}

Taking the limit as $h \to 0$, we get
\[
\lim_{h \to 0}  \mathbb{E}\left[\left(\frac{X(t+h) - X(t)}{h}\right)^2\right]  = \lim_{h \to 0} \frac{\lambda}{4 |h|} = +\infty,
\]
and conclude that the stochastic process $X(t)$ is not mean square differentiable. 


\section{Noisy measurements of jump times}
\label{section:jump}

The jump times of a trajectory $\mathbf{X}_i \in D([0,1], \mathbb{R}^q)$
are defined as follows
\begin{align}
J_i = \{ t \in [0,1] ; \mathbf{X}(t-) \neq   \mathbf{X}(t) \}
\end{align}
The cardinal of $J_i$ is finite, say $n_i$ and we denote by $\{t_{i1} < \ldots   < t_{in_i} \}$ its elements ordered in ascending order.  
As suggested by an anonymous referee, one could consider that the observed  jump times are corrupted by noise, so that we effectively observe the jumps at the noisy time instants 
\[
\widetilde{t}_{ik} = t_{ik} + \epsilon_{ik}, \quad k=1, \ldots, n_i,
\]
with the following correction on the edges of the interval, $\widetilde{t}_{ik} = 0$ if $t_{ik} + \epsilon_{ik} <0$ and $\widetilde{t}_{ik} = 1$ if $t_{ik} + \epsilon_{ik} >1$.

We have evaluated the influence of such noisy observations on the results of our analysis, considering iid noise $\epsilon_{ik}$, with uniform distribution on $[-a,a]$ with a = 0.05, 0.1 and 0.2. 
Noise was added on the jump instants of TDS data presented in Figure~\ref{fig:TDS} and MFPCA, using the same parameters as previously (8 B-splines of order 3, with equispaced knots and time discretized in 100 equispaced points in $[0,1]$). A hundred simulations were performed for each scenario.

The average absolute correlations between the individual principal component scores of the first two components and those of the original (noise-free) first components were computed and are presented in Figure~\ref{fig:noisy}.
As expected, when the noise variance is small ($a = 0.05$), the correlations remain very close to one between the original scores and those based on noisy data. As it increases ($a = 0.1$ and $a = 0.2$), the correlation decreases, unsurprisingly, with a more pronounced reduction for the second component than for the first.

\begin{figure}
\centering
\includegraphics[scale=.65]{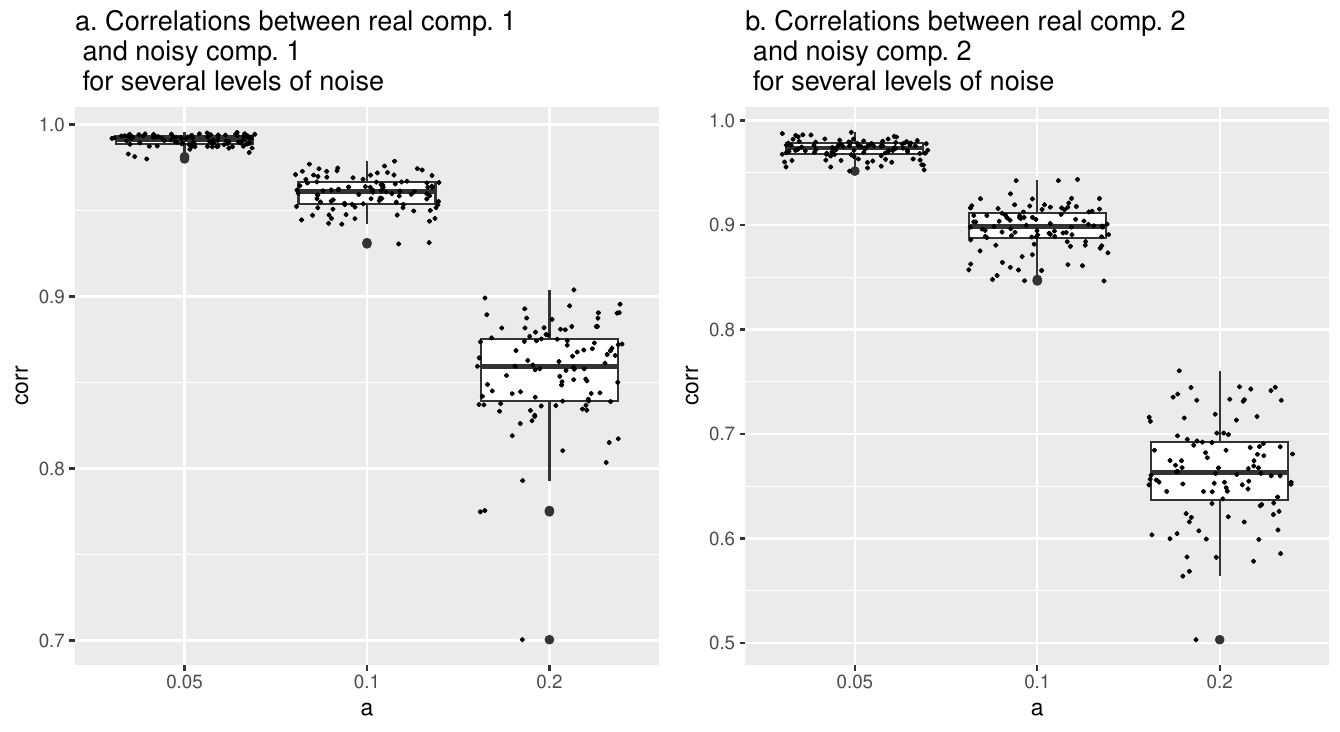}
\caption{Distribution of the correlation coefficient between original and noisy first two principal components with time jumps corrupted by an additive noise with uniform distribution over $[-a,a]$,  for $a \in \{ 0.05, 0.1, 0.2 \}$.}
\label{fig:noisy}       
\end{figure}

\end{document}